\documentclass{article}

\usepackage[preprint]{neurips_2026}

\usepackage[utf8]{inputenc} 
\usepackage[T1]{fontenc}    
\usepackage{hyperref}       
\usepackage{url}            
\usepackage{booktabs}       
\usepackage{amsfonts}       
\usepackage{nicefrac}       
\usepackage{microtype}      
\usepackage{xcolor}         

\usepackage{amsmath,amssymb,mathtools,amsthm}
\usepackage{thmtools}
\usepackage[capitalise]{cleveref}

\providecommand{\R}{\mathbb{R}}     

\providecommand{\E}{\mathbb{E}}     
\DeclareMathOperator{\sign}{sign}   

\DeclarePairedDelimiter{\abs}{\lvert}{\rvert}

\DeclarePairedDelimiterX{\inner}[2]{\langle}{\rangle}{#1,\,#2}

\theoremstyle{plain}
\newtheorem{theorem}{Theorem}[section] 

\newtheorem{proposition}[theorem]{Proposition}

\newtheorem{lemma}[theorem]{Lemma}

\theoremstyle{definition}
\newtheorem{definition}[theorem]{Definition}

\theoremstyle{remark}
\newtheorem{remark}[theorem]{Remark}
\crefname{theorem}{Theorem}{Theorems}
\Crefname{theorem}{Theorem}{Theorems}

\crefname{definition}{Definition}{Definitions}
\Crefname{definition}{Definition}{Definitions}
\crefalias{definition}{definition}

\crefname{proposition}{Proposition}{Propositions}
\Crefname{proposition}{Proposition}{Propositions}
\crefalias{proposition}{proposition}

\crefname{lemma}{Lemma}{Lemmas}
\Crefname{lemma}{Lemma}{Lemmas}
\crefalias{lemma}{lemma}

\crefname{corollary}{Corollary}{Corollaries}
\Crefname{corollary}{Corollary}{Corollaries}
\crefalias{corollary}{corollary}

\crefname{remark}{Remark}{Remarks}
\Crefname{remark}{Remark}{Remarks}
\crefalias{remark}{remark}

\crefname{algorithm}{Decision Tree}{Decision Trees}
\Crefname{algorithm}{Decision Tree}{Decision Trees}
\crefalias{algorithm}{algorithm}

\usepackage{float}          
\usepackage{subcaption}     
\usepackage{thm-restate}
\usepackage{algorithm}
\floatname{algorithm}{Decision tree}
\usepackage{algpseudocode}

\usepackage{enumitem}   
\usepackage{makecell}
\usepackage{multirow}

\crefname{equation}{Equation}{Equations}
\crefformat{equation}{Equation~#2#1#3}   
\Crefformat{equation}{Equation~#2#1#3}   
\crefname{figure}{Figure}{Figures}
\Crefname{figure}{Figure}{Figures}
\crefname{table}{Table}{Tables}
\Crefname{table}{Table}{Tables}

\theoremstyle{definition}
\newtheorem{assumption}{Assumption}[section]
\crefname{assumption}{Assumption}{Assumptions}
\Crefname{assumption}{Assumption}{Assumptions}

\title{Do covariates explain why these groups differ? The choice of reference group can reverse conclusions in the Oaxaca-Blinder decomposition}

\author{%
  Manuel Quintero\thanks{Equal contribution.} \\
  Massachusetts Institute of Technology, Cambridge, MA, USA
  \And
  Advik Shreekumar$^*$ \\
  University of California, Berkeley, Berkeley, CA, USA
  \And
  William T.~Stephenson \\
  Massachusetts Institute of Technology, Cambridge, MA, USA
  \And 
  Tamara Broderick \\
  Massachusetts Institute of Technology, Cambridge, MA, USA
}

\newcommand\blfootnote[1]{%
    \begingroup
    \renewcommand\thefootnote{}\footnote{#1}%
    \addtocounter{footnote}{-1}%
    \endgroup
  }

\begin{document}

\maketitle
\blfootnote{DISTRIBUTION STATEMENT A. Approved for public release. Distribution is unlimited.}

\begin{abstract}
Scientists often want to explain why an outcome is different in two groups. For instance, differences in patient mortality rates across two hospitals could be due to differences in the patients themselves (covariates) or differences in medical care (outcomes given covariates). The Oaxaca–Blinder decomposition (OBD) is a standard tool to tease apart these factors. It is well known that the OBD requires choosing one of the groups as a reference, and the numerical answer can vary with the reference. To the best of our knowledge, there has been no systematic investigation into whether the choice of OBD reference can yield different substantive conclusions and how common this issue is. In the present paper, we give existence proofs in real and simulated data that the OBD references can in fact yield substantively different conclusions. Our empirical exercises find that this sensitivity is more common when the OBD is extended to more complex regression models, including a pretrained transformer. Our theoretical and empirical results together establish that these conclusion reversals are not entirely driven by model misspecification, small data, or adversarial parameter choices. Our results suggest that practitioners should always report both directions of the OBD; that modern machine learning and large datasets do not automatically resolve the conclusion reversal problem; and that further work on this problem is needed.
\end{abstract}

\section{Introduction}
\label{sec:intro}
As a motivating example, suppose a data scientist at Hospital $H$ notices that mortality rates are noticeably higher at Hospital $H$ than at Hospital $K$.
She knows the medical literature links many risk factors ($X$), such as high blood pressure, to mortality ($Y$).
The observed differences in mortality could be due to differences in these risk factors between the two groups of patients (distribution of $X$) or to differences in how well each hospital manages these risk factors (distribution of $Y \mid X$).
If she can figure out how much each of these components is driving the difference in mortality rates, her hospital may be able to change its policy to improve patient care.
How can she use the available information to form hypotheses about what drives the difference in mortality?

The modern explainability literature in machine learning boasts a bevy of methods for attributing outcome behavior to covariates \citep{Dwivedi2023}.
These include Shapley values \citep{Shapley1953, Lundberg2017, chen2025shapley}, Functional ANOVA \citep{stone1994, hooker2004discovering, hookerGeneralizedFunctionalANOVA2007, lengerich2020purifying, fumagalli2025fanova}, Accumulated Local Effects \citep{apley2020}, and partial dependence plots \citep{friedman2001greedy, liu2025trees}.
However, such methods are not designed to answer this data scientist's question.
They attribute variation in outcomes within a single population to the covariates, holding the conditional distribution $Y \mid X$ fixed.
But, the data scientist is interested in why an outcome varies between two populations, where the $Y \mid X$ may vary.
See \cref{appendix:xAI} for more discussion on why these methods do not solve the data scientist's question.

Instead, the data scientist might adapt the Oaxaca--Blinder decomposition (OBD), which attributes disparities between two groups to $X$ and $Y \mid X$. 
The OBD was developed for linear models \citep{oaxaca1973,blinder1973} and remains widely used in this form --- e.g., to analyze gender wage gaps \citep{paradaContzen2025}, educational attainment differences \citep{barhaim2023}, and health disparities \citep{cartwright2021}. Recent work has explored pairing the OBD with flexible functional decompositions \citep{quintero2025,bach2024heterogeneity,quintas2024multiply,mao2025double,flachaire2025decomposing}; see \cref{appendix:obd_generalized} for more discussion of these methods.

Since the original paper \citep{oaxaca1973}, it has been well known that the OBD requires choosing one of the two groups (Hospital $K$ or $H$ in our example) as a reference, and that this choice affects the numerical result.
This dependence on the reference choice is known as the ``index number problem.''
The index number problem would be less concerning if both references always yield the same substantive conclusions;
in our example, if both indicate that differences in the distribution of risk factors widen the mortality gap between the two hospitals.
But if one reference choice finds that risk factors widens the gap while the other finds that they narrow it, the choice of reference \emph{would} matter for qualitative conclusions and subsequent decision making.

We are not aware of an investigation into whether the choice of reference groups can yield different \emph{substantive} conclusions, either
for the original OBD with linear regression or for recent extensions to more complex, nonlinear regression models.
In the present work, after reviewing the OBD and the index number problem (\cref{sec:counterfactual_OBD}), we 
give an existence proof (with real data) that the OBD \emph{can} yield two substantively different conclusions (\cref{sec:real_example}).\footnote{We release all of our code to replicate the experiments at the following repository, \url{https://github.com/manuelquinteroc/ReferenceOaxacaBlinder}.}
We analyze a range of data (healthcare and U.S.\ Census data), and we show that the problem persists for the more complex, nonlinear extensions of the OBD.
In fact, we show that these extensions, including one with a recent foundation model, seem to suffer much more frequently from the issue than
a simple linear model.
In \cref{sec:signflips}, we further demonstrate theoretically that this phenomenon is not an artifact of 
model misspecification, small data size, or adversarial parameter choices.
We suggest best practices for data analysts in \cref{sec:conclusion}.

\section{Setup and review of the Oaxaca--Blinder decomposition (OBD)}
\label{sec:counterfactual_OBD}

We first review the index number problem and OBD, defining them at the population level where we have full information about all variables of interest.
Then we describe how to proceed with real data.


\subsection{Differences in means as a counterfactual exercise} 
\label{sec:counterfactual_example}

As a first step toward describing the OBD, we illustrate the index number problem by framing it as a counterfactual exercise as in \citet{fortinLemieuxFirpo2011}.
We consider an observed outcome $Y \in \mathbb{R}$ and covariates $X$ in some space $\mathcal{X}$. We take two groups, $H$ and $K$. We assume that the tuple $(X,Y)$ is independently and identically distributed within each group $g$ with corresponding expectation $\E_g$. We suppose we observe a gap in the mean outcomes, $\E_H[Y]$ and $\E_K[Y]$, and would like to understand to what extent the gap is due to differences in the distribution of $X$ versus differences in the distribution of $Y\mid X$ between the two groups. 
For the moment, we assume we have access to population-level quantities, including all expectations.

To analyze the gap between $\E_H[Y]$ and $\E_K[Y]$, we might construct a counterfactual driven only by differences in covariates.
Again taking our example, where $Y$ represents mortality and $X$ represents risk factors, we might ask:
(A) \emph{How would Hospital $H$'s mortality rate change if its patients had the risk factor distribution of patients at Hospital $K$?}
But, this question starts from the patient group at Hospital $H$ as a reference. We could instead start from group $K$ and ask:
(B) \emph{How would Hospital $K$'s mortality rate change if its patients had the risk factor distribution of patients at Hospital $H$?}


More generally, we might ask:
(A) \emph{How much would Group $H$'s mean outcome ($Y$) change if Group $H$ kept its outcome--covariate relationship ($Y \mid X$) but had the covariate ($X$) distribution of Group $K$?}
or 
(B) \emph{How much would Group $K$'s mean outcome ($Y$) change if Group $K$ kept its outcome--covariate relationship ($Y \mid X$) but had the covariate ($X$) distribution of Group $H$?}

To address the general questions, it will be useful to compute the mean of a counterfactual $Y$ that has the outcome--covariate relationship ($Y \mid X$) of group $g$ and the covariate ($X$) distribution of group $g'$. By the law of total expectation, this mean equals
$\E_{g'}[\E_g[Y \mid X]]$.
Then the difference in mean outcomes between the two groups can be decomposed into differences between each group's mean
and the chosen counterfactual mean. The decompositions corresponding to questions (A) and (B) above are, respectively, 
\cref{eq:OB_H_Counterfactual,eq:OB_K_Counterfactual} below.
\begin{align}
\Delta_Y := \E_H[Y] - \E_K[Y] 
    \label{eq:OB_H_Counterfactual}
    &=\underbrace{\E_H[Y] - \E_K[\E_H[Y \mid X]] }_{E^{(H)}} 
     + \underbrace{\E_K[\E_H[Y \mid X]] - \E_K[Y]}_{U^{(H)} } \\
     \label{eq:OB_K_Counterfactual}
     &=\underbrace{\E_H[Y] - \E_H[\E_K[Y \mid X]] }_{U^{(K)}} 
     + \underbrace{\E_H[\E_K[Y \mid X]] - \E_K[Y] }_{E^{(K)}}.
\end{align} 
%
We follow the traditional naming and notation for these terms \citep{fortinLemieuxFirpo2011}.
The term corresponding to changing the covariates is called the
``explained component'' since this part of the difference is explained by the observed
covariates. The remaining term, corresponding to changing the conditional $Y \mid X$ distribution,
is called the ``unexplained component'' since it is not explained by the observed covariates. We denote these components as $E$ and $U$, respectively, with superscript denoting the reference group choice.

We might interpret $E$ and $U$ as the respective contributions of the covariates and outcome--covariate relationship to the 
difference in outcome means between groups $H$ and $K$. The \textbf{index number problem}, then, represents the observation
that in general, $E^{(K)} \neq E^{(H)}$ and $U^{(K)} \neq U^{(H)}$.

\textbf{A conceptual resolution?}
Before proceeding, we observe that there is not an immediate and general conceptual resolution of the index number problem.
It is not clear to us in the hospital example above that one direction is more natural, or even that the two directions represent distinct policy questions.

Generic responses to the index number problem do not resolve it without relying on arbitrary choices or compromising the interpretation of results; see \cref{sec:existing_index_number} for detailed discussion, which we summarize here.
First, approaches like those of \citet{oaxacaRansom1998,neumark1988} define a new reference group by aggregating $H$ and $K$. However, there are many ways to perform this aggregation, introducing a new, still-arbitrary choice.
Second, use of Shapley values \citep{Shapley1953} could call for averaging the explained or unexplained components for each reference group.
However, we show that the choice of reference group can reverse the sign of either component, complicating their interpretation. 
Shapley-style approaches thus average two questionable quantities and obscure information when the averaged components have opposite signs.

In practice, we do not see a consensus approach, even within fields of study.
For instance, in economic analyses of wage gaps, some argue that a single reference group makes the most sense to use in context, as in \citet{dinardoFortinLemieux1996}'s study of unionization, the minimum wage, and economic inequality.
Often, the choice is a single reference group that aggregates the two original references in some way \citep{oaxacaRansom1994,reimers1983,cotton1988,kassenboehmer2014distributional,sharafRashad2016,rahimiNazari2021,Allen2022,bachan2022genderchancellor}.
However, we also observe published work that takes either men or women as the reference group \citep{nielsen2000wage,piazzalunga2019increase,leythienne2021paygapeu,tao2024gender,paradaContzen2025}.
Some papers report results for both references \citep{ONeill2006,fortinLemieuxFirpo2011,sharafRashad2016,rahimiNazari2021}. 
Still others report results for just one reference, with no argument for excluding the other reference \citep{zhang_2019, ayubiShahbaziKhazaei2024,paradaContzen2025,singleton2016,lamJamiesonMittinty2021,leythienne2021paygapeu,alHanawiNjagi2022,tao2024gender}.

\subsection{Real-data Oaxaca--Blinder decomposition and drawing conclusions} 
\label{sec:ob_real}
In practice, researchers do not have access to population quantities and must instead draw conclusions from finite-data estimates.
In particular, implementing \cref{eq:OB_H_Counterfactual,eq:OB_K_Counterfactual} requires estimating $\E_g[Y]$, typically with the sample mean, and the counterfactual mean $\E_{g'}[\E_g[Y \mid X]]$, which requires more thought. 
Practitioners typically estimate counterfactual means by (1) using data from $g$ to compute an estimator of $\E_g[Y \mid X]$ (essentially a regression), and (2) using data from $g'$ to take the sample average across the estimator from the first step.
Concretely, for each $g \in {\{H, K\}}$, suppose we have $N_g$ observations of $\{X_n, Y_n\}$. Let $\hat{M}_g(X)$ be an estimator for $\E_g[Y \mid X]$.
Then the following equation implements \cref{eq:OB_H_Counterfactual}:
\begin{align}	
    \hat{\Delta}_Y
    \label{eq:NOBD_H_estimated}
    &=\underbrace{ \frac{1}{N_H} \sum_{n=1}^{N_H} Y_n - \frac{1}{N_K} \sum_{n=1}^{N_K} \hat{M}_H(X_n)}_{\hat{E}^{(H)}} 
     + \underbrace{\frac{1}{N_K} \sum_{n=1}^{N_K} \hat{M}_H(X_n) - \frac{1}{N_H} \sum_{n=1}^{N_H} Y_n}_{\hat{U}^{(H)} }.
\end{align}

In the literature, the term \emph{OBD} historically refers to using \cref{eq:OB_H_Counterfactual,eq:OB_K_Counterfactual} with the further choice of linear regression to compute $\hat{M}_g(X)$.
We emphasize that the use of linear regression is a modeling choice, and that it is possible to choose nonlinear estimators for $\E_g  [ Y | X ]$.
For instance, \citet{quintas2024multiply,flachaire2025decomposing,mao2025double} explore generalizations of the OBD that use machine learning methods to estimate $\E_g  [ Y | X ]$, and derive asymptotic properties for such procedures.
To distinguish from the standard OBD, we will refer to the case where $\hat{M}_g(X)$ is a nonlinear estimator as the \emph{nonlinear Oaxaca-Blinder decomposition} (NOBD).

\textbf{Drawing conclusions from the (N)OBD.}
It is common to draw substantive conclusions based on the signs (and statistical significance) of the empirical explained and unexplained components \citep{oaxacaRansom1998, jann2008,fortinLemieuxFirpo2011}. Returning to the hospital example from above, suppose Hospital $H$ has lower mortality ($\Delta_Y < 0$), and serves as the reference group. 
If the (N)OBD explained component has a negative sign, we might conclude that $H$ can attribute
its lower mortality rate in part to patient characteristics, such as healthier blood pressure. 
Analogously, if the (N)OBD unexplained component has a significant negative sign, we might attribute group $H$'s lower
mortality in part to the quality of hospital care.
In either case, we might also require that the observed components be statistically distinguishable from zero to form a conclusion; without significance, we might hesitate to attribute differences in mortality to covariates or care quality. 
For the rest of the paper, we will consider a practitioner drawing their conclusions based on the (statistically significant) sign of the explained or unexplained component.

\section{Real-data example: Gender gap in ICU data}
\label{sec:real_example}
The index number problem reflects that the two reference choices for the (N)OBD could yield two different sets of values for the explained and unexplained components. We are not aware of an illustration in the literature that the two directions (for either the OBD or NOBD) can lead to \emph{substantively} different conclusions. We supply such an illustration next for a real-data analysis in healthcare.

\textbf{Data.} To construct an example, we examine data from the PhysioNet cohort of intensive care unit (ICU) patients collected for a study of in-hospital mortality \citep{goldberger2000physiobank}, where the outcome $Y$ is binary, indicating whether the patient died in the hospital or not. The dataset contains routine measurements taken at admission, including heart rate (HR), mean arterial pressure, temperature, and urine output. These values are recorded for every patient at the time of entry into the ICU and provide a description of their initial physiological state. To simulate various real-data analyses where practitioners might have access to only a subset of these patients, we construct $139$ versions of this dataset with various patient subsets; see \cref{appendix:icu_details}.
For reference groups, we consider males and females.


\textbf{Analysis.}
As discussed in \cref{sec:ob_real}, practitioners often draw conclusions based on the sign and statistical significance of the (N)OBD explained and unexplained components.
To explore whether these conclusions can be robust to the choice of reference group, we compute the explained and unexplained components for each of our $139$ data subsets across one foundation model and four classic models. The foundation model is TabPFN, a $25$ million parameter, transformer-based model, pre-trained to map training datasets to predictive distributions $Y|X$ \citep{hollmann2023tabpfn}. The other four models are linear regression,\footnote{Despite the binary outcome $Y$, it is common practice in applied health and policy research to fit separate linear models (via ordinary least squares) for the relationship between a binary outcome and covariates in each group, and then apply the OBD \citep{jann2008, edokaChangesCatastrophicHealth2017, sujin_LPM_OB, mweembaGapSelfRatedHealth2023}. In fact, this use of the linear model for a binary outcome is so common that it has a name, the ``linear probability model.''} logistic regression, XGBoost \citep{chen2016xgboost}, and a two-layer fully-connected neural network.
We assess statistical significance using $1{,}000$ bootstrap replications.
\cref{tab:flip_summary_main} counts the number of sign flips across these models at various levels of statistical significance.
Concerningly, we see that \emph{every} model contains at least one sign flip, and more complex models tend to have even more flips; for TabPFN, we see that more than  half of the datasets show opposite signs for at least one of the explained or unexplained components, and $13.7\%$ show significant sign flips at the $5\%$ level.


\begin{table}[!ht]
    \caption{Number of data subsets (out of $139$) exhibiting at least one sign flip (in either the explained or unexplained component), by model. Columns $10\%$, $5\%$, and $1\%$ report the number of flipped subsets that are significant at the corresponding level.}
    \label{tab:flip_summary_main}
    \centering
    \small
    \begin{tabular}{lcccc}
    \toprule
    Model & Total flips & $10\%$ & $5\%$ & $1\%$ \\
    \midrule
    Linear      & 27 &  2 &  2 &  2 \\
    Logistic    & 29 &  2 &  2 &  1 \\
    Neural net  & 62 & 29 & 17 &  6 \\
    XGBoost     & 62 & 17 & 10 & 3 \\
    TabPFN     & 76 & 22 & 19 & 5 \\
    \bottomrule
    \end{tabular}
\end{table}

\subsection{A deeper dive on one data subset}

\cref{tab:flip_summary_main} shows that sign flips are common.
But what do these flips really mean in practice?
And are they merely an artifact of noise or instead a structural phenomenon?
As is often the case in studying complex machine learning models, studying linear regression provides a tractable proxy problem.
In the present study, linear regression is also of direct interest since the OBD (with linear regression) is widely used in practice.
So, for the rest of the paper, we focus on \cref{eq:OB_H_Counterfactual,eq:OB_K_Counterfactual} with linear regression (i.e., the OBD instead of NOBD).
To start answering our questions, we examine one particular data subset leading to a significant sign flip for linear regression in \cref{tab:flip_summary_main}.
We will see what it means for sign flips to change meaningful conclusions and also see that flips have a natural structural interpretation.

The particular data subset we consider covers patients with an elevated heart rate.
Medical practitioners often look at moderate elevations in heart rate as a first indication of stress, infection, or concerning changes in blood flow and circulation. In the present example, we consider patients whose admission heart rate lies between the $50$th and $75$th percentiles. This patient set, which we refer to as HR quartile two, reflects a clinical presentation in which the heart rate is elevated but not extreme.

\textbf{Two different conclusions.} 
We report the OBD for each reference group in \cref{tab:hr_q2_decomposition}. We next interpret the substantive conclusions
we might draw from each reference group. First, consider women as the reference. The explained component is positive ($0.021$), with a statistically significant difference from zero.
This result suggests that differences in admission covariates place men at higher mortality risk. A medical practitioner might therefore examine whether various factors before admission differ across gender; these factors include early differences in blood flow and circulation, underlying disease burden, or overall severity at admission.

Second, consider men as the reference. The explained component is negative ($-0.007$), though the difference from zero is not statistically significant. In this case, the observed covariate differences seem to predict \emph{lower} (or equal) mortality for men relative to women rather than higher. In other words, the same covariates would now be interpreted as offering men a slight mortality advantage. This interpretation is substantively different than the interpretation that results from using women as the reference.


\begin{table}[!ht]
  \caption{Oaxaca--Blinder decompositions for HR quartile two. The total gap is men's mean minus women's mean. We compute standard errors (in parentheses) via bootstrap resampling with $1{,}000$ replications. 
    We calculate p-values using a normal approximation (Wald test), where the test statistic is the coefficient divided by its bootstrap standard error. 
    Significance levels: * $p<0.10$, ** $p<0.05$, *** $p<0.01$.}
  \label{tab:hr_q2_decomposition}
  \begin{center}
    \begin{small}
        \begin{tabular}{lcccr}
          \toprule
          Reference & Explained & Unexplained & Total gap  \\
          \midrule
              Women 
    & \makecell{$0.021^{***}$ \\ $(0.007)$}
    & \makecell{$0.014$ \\ $(0.025)$}
    & \makecell{$0.035$ \\ $(0.023)$}
    \\
    Men 
    & \makecell{$-0.007$ \\ $(0.011)$}
    & \makecell{$0.042^{*}$ \\ $(0.025)$}
    & \makecell{$0.035$ \\ $(0.023)$}
    \\
          \bottomrule
        \end{tabular}
    \end{small}
  \end{center}
  \vskip -0.1in
\end{table}

\textbf{How does the sign flip arise?}
To understand how the sign flip arises, it will help to introduce more notation for the linear regression case.
We now restrict to real-valued covariates $X \in \mathbb{R}^d$, and we assume we fit a linear regression to each group $g$: $\alpha_g + X^\top \beta_g$
with intercept $\alpha_g \in \mathbb{R}$ and slope $\beta_g \in \mathbb{R}^d$. We let
$\mu_g := \E_g[X],
\Delta \mu := \mu_H - \mu_K,
\Delta \beta := \beta_H - \beta_K,
\Delta \alpha := \alpha_H - \alpha_K,
$
and use a hat to denote estimators of these quantities (cf.\ \cref{sec:ob_real}).

With these choices, the OBD
(\cref{eq:OB_H_Counterfactual,eq:OB_K_Counterfactual}) becomes
\begin{equation}
    \Delta_Y = \underbrace{\Delta \mu^\top \beta_H}_{E^{(H)}} + \underbrace{\mu_K^\top \Delta \beta + \Delta \alpha}_{U^{(H)}} = \underbrace{\Delta \mu^\top \beta_K}_{E^{(K)}} + \underbrace{\mu_H^\top \Delta \beta + \Delta \alpha}_{U^{(K)}}.
\end{equation}
The first equality is the OBD with reference group $H$, and the second with reference group $K$.

With this notation in hand, we observe that $\Delta\hat{\mu}$ is fixed between $\hat{E}^{(H)}$ and $\hat{E}^{(K)}$. So it must be that some components of $\hat{\beta}_H$ and $\hat{\beta}_K$ are changing sign to give $\hat{E}^{(H)} := \Delta\hat{\mu}^T \hat{\beta}_H$ and $\hat{E}^{(K)} := \Delta\hat{\mu}^T \hat{\beta}_K$ opposite signs.
In \cref{fig:sign_flip_overview} (left), we show the fitted linear models of in-hospital mortality on admission heart rate, holding all other covariates fixed at their group means.
We see that the signs of the slopes differ.
In Appendix \cref{tab:delta_mu_betas}, we report values for all elements of each vector $\hat{\beta}_{g}$ in the full analysis. 
We see that there are multiple covariates $c$ for which $\hat{\beta}_{\textrm{men},c} < 0 < \hat{\beta}_{\textrm{women},c}$.
And as $\Delta\hat{\mu}_c$ has the same sign for all of these covariates, each one pushes towards a sign flip. 
In aggregate, these changes add enough to cause a sign flip in the explained component.

\begin{figure}[!ht]
    \centering
    \begin{subfigure}[t]{0.48\linewidth}
        \centering
        \includegraphics[width=\linewidth]{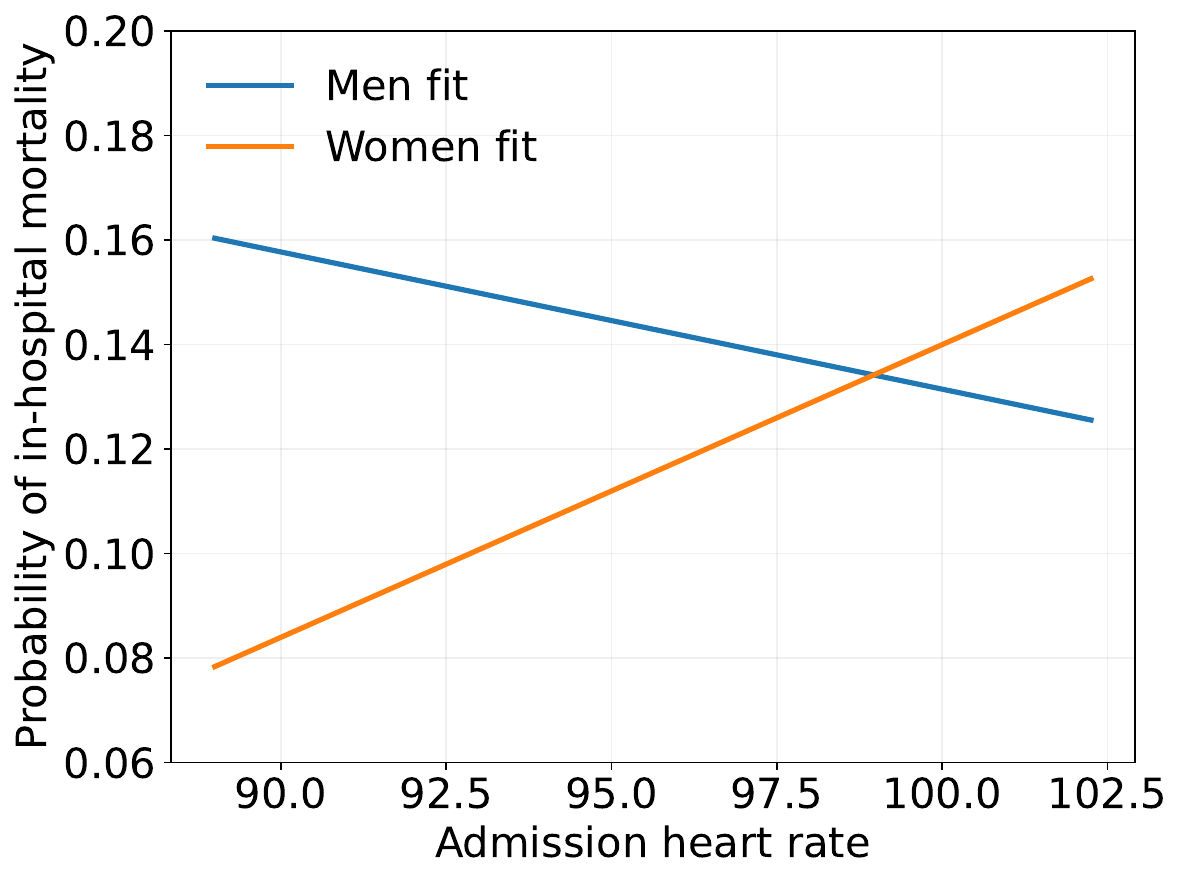}
        \label{fig:hr_quartile2}
    \end{subfigure}
    \hfill
    \begin{subfigure}[t]{0.48\linewidth}
        \centering
        \includegraphics[width=\linewidth]{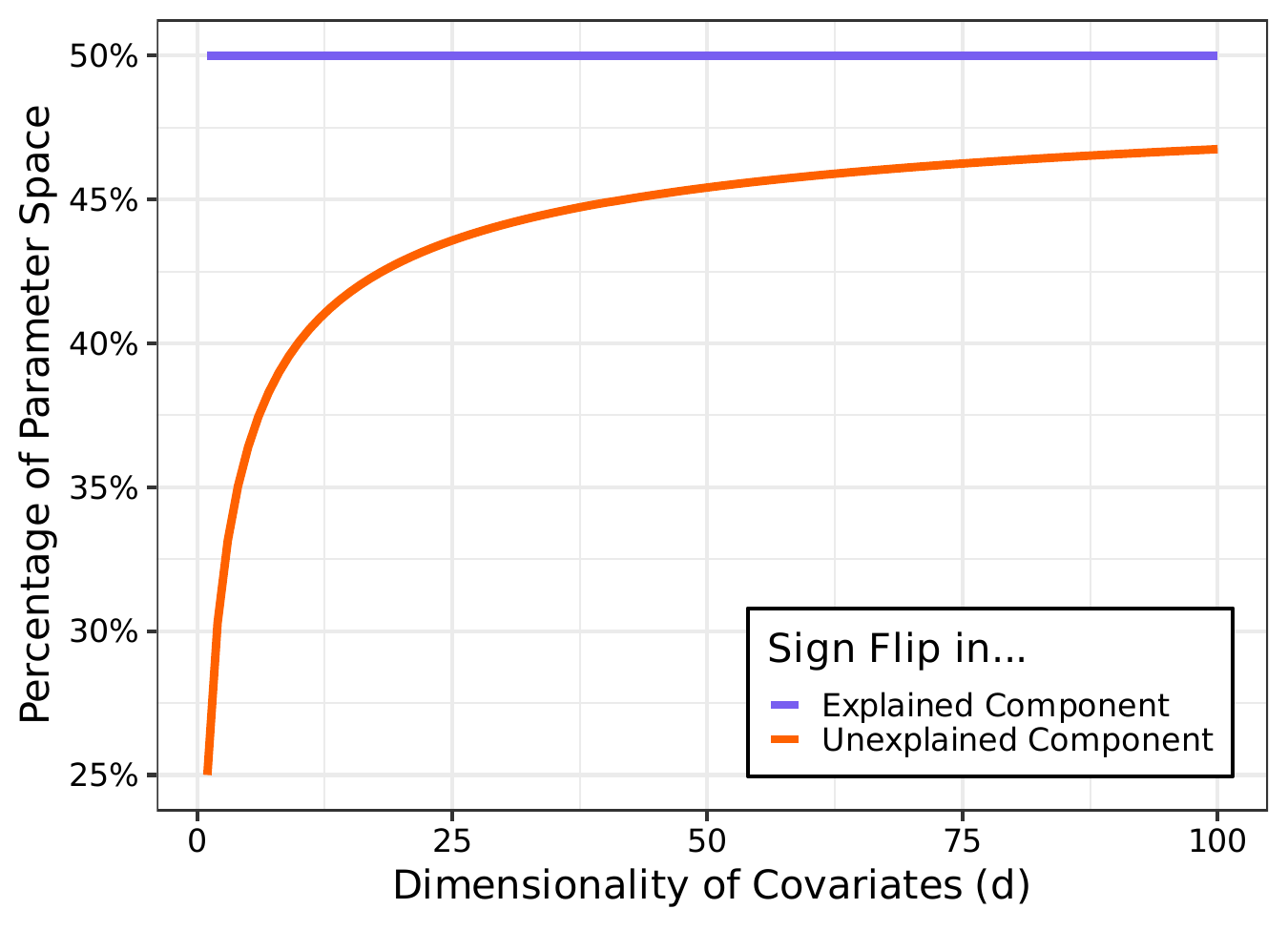}
        \label{fig:prob_unexplained}
    \end{subfigure}
    \caption{Left: Linear fit of mortality on admission heart rate for each group when holding all other covariates fixed at their group means. Right: Percentage of parameter space leading to sign flips. See \cref{remark:computing_irwin_hall} for how we compute the percentage for the unexplained component.}
    \label{fig:sign_flip_overview}
\end{figure}

\textbf{Limitations.}
Our analysis reveals sign flips in the explained and unexplained components of the (N)OBD in a real dataset.
It is a priori possible that some of our choices may be unusual; these choices include the dataset itself, our choice of covariates, our choice of outcome, or our particular subset of data. And it is possible that sign flips may not occur in other settings.
To address this possibility, we provide a second illustration analyzing differences in income as well as health insurance status using U.S.\ Census data in \cref{appendix:us_labor_force}.
The latter data have previously been analyzed by \citet{bach2024heterogeneity}, reducing the concern that our empirical results are due to atypical choices on our part.
We also give a theoretical account of sign flips in \cref{sec:signflips}, where we establish that they correspond to a nontrivial region of the parameter space.
\section{How common are sign flips?} 
\label{sec:signflips}

Our examples above serve as an existence proof that reference group choice can reverse the (N)OBD's substantive conclusions by flipping signs of the explained or unexplained components.
However, a few major questions remain.
First, do these conclusion reversals occur in the absence of finite-sample noise; i.e., do practitioners with large-scale datasets still need to worry about them?
Second, even if conclusion reversals can occur, do they occur outside of mathematically adversarial examples (i.e., ``measure zero'' in some sense)?

In this section, we mathematically prove the answer to both questions is yes.
\cref{sec:simulated_example} complements our argument with an example where a sign flip occurs even without sampling noise.

In what follows, we study conclusion reversals that involve the explained or unexplained component changing signs from positive to negative, or vice versa.
For mathematical tractability and applicability to practical usage, we continue our focus on the OBD (vs.\ NOBD). 
To show that sign flipping is not an artifact of finite samples, we work in the population limit; since we work in the population limit, significance is no longer a concern.
And to show that sign flipping is not a ``measure zero'' phenomenon, we show that sign flips of both components represent a large proportion of the parameter space, in a sense we make precise below.

\subsection{Characterizing sign flips}
\label{sec:sign_flip_characterizations}

We start by characterizing sign flips in the explained and unexplained components, respectively; these characterizations will help us later establish
that sign flips of either component are common in parameter space.

\textbf{Explained component.} Explained component sign flips are relatively straightforward to characterize.
\begin{definition}
	\label{defn:sign_flip_explained}
	We say the OBD \emph{explained component flips signs} if the sign of the explained term, $\Delta\mu^T \beta_g$, depends on the reference group. That is,
	\begin{equation}
        \label{eq:sign_flip_condition_explained_neq}
		\sign\!\left(\Delta\mu^T \beta_H\right) \neq \sign\!\left(\Delta\mu^T \beta_K\right).
	\end{equation}
\end{definition}

From the definition, we can immediately see the following alternative characterization.
\begin{remark}
    \label{rmk:explained_ordering_sensitivity}
    The OBD explained component flips signs if and only if the following two conditions hold:
    \begin{equation}
    \label{eq:sign_flip_condition_explained}
        \Delta\mu^T \beta_H \neq \Delta\mu^T \beta_K, \quad \min\{\Delta\mu^T \beta_H,\ \Delta\mu^T \beta_K\} < 0 < \max\{\Delta\mu^T \beta_H,\ \Delta\mu^T \beta_K\}.
    \end{equation}
\end{remark}
Since $\Delta\mu$ is constant across analyses, sign flips in the explained component are straightforwardly driven
by differences in sign in the slope between groups.



\textbf{Unexplained component.}
Unlike the explained component, sign reversals in the unexplained component arise from the interaction between slope and intercept differences.
\begin{definition}
    \label{def:unexplained_sign_flip}
	Recall $\Delta\beta := \beta_H - \beta_K$ and $\Delta\alpha := \alpha_H - \alpha_K$. The OBD \emph{unexplained component flips signs} if the sign of the unexplained term, $\mu_{g'}^T \Delta\beta + \Delta\alpha$, depends on the reference group. That is
	\begin{equation}
		\sign\left( \mu_K^T \Delta\beta + \Delta\alpha \right) \neq \sign\left( \mu_H^T \Delta\beta + \Delta\alpha \right).
	\end{equation}
\end{definition}

The following alternative characterization for sign flips of the unexplained component is analogous to \cref{rmk:explained_ordering_sensitivity}.
\begin{restatable}[OBD unexplained sign flips]{proposition}{OBSignFlips}
    \label{thm:unexplained_ordering_sensitivity}
    The OBD unexplained component flips signs if and only if the following two conditions hold:
    \begin{equation}
    \label{eq:sign_flip_condition}
        \mu_H^T \Delta\beta \neq \mu_K^T \Delta\beta, \quad \min\{\mu_H^T \Delta\beta, \mu_K^T \Delta\beta\} < -\Delta\alpha < \max\{\mu_H^T \Delta\beta, \mu_K^T \Delta\beta\}.
    \end{equation}
\end{restatable}    
See \cref{sec:sign_flip_proof} for the proof.
From \cref{thm:unexplained_ordering_sensitivity}, we see that group-specific slopes serve a different role in sign flips in the unexplained case relative to the explained case.
Suppose we hold fixed the covariate means $\mu_H$ and $\mu_K$. Then scaling up the magnitude of the vector change-in-slopes ($\Delta \beta$) induces sign flips for a greater range of values 
of the scalar change-in-intercepts ($\Delta \alpha$).
That is, larger magnitudes of \( \Delta \beta \) create ``more room'' for sign flips in the unexplained component.

\subsection{Are sign flips an adversarial phenomenon?}
\label{sec:sign_flip_fraction}

After characterizing sign flips in \cref{sec:sign_flip_characterizations}, we can ask whether they occur only for adversarially chosen parameters.
One formalization of this question is to ask what proportion of the volume of parameter space corresponds to sign flips --- and in particular, whether that proportion is non-zero.

An immediate challenge with talking about volume within a parameter space is that the parameter space inherits
its units from the covariates and response; in particular, any individual unit choice is somewhat
arbitrary (e.g., centimeters vs.\ millimeters), but also different covariates need not be 
comparable (e.g., temperature vs.\ weight). 
The usual way to alleviate this incompatibility in one group is to center the covariates and set their scale (e.g., standard deviation) to be 1. 
We can use similar normalization to center the covariates such that $\mu_K = \mathbf{0} \in \R^d$ and scale the covariates such that $\Delta\mu = \mathbf{1} \in \R^d$ (i.e., $\mu_H = \mathbf{1}$).
Our standardization choice ($\mu_K= \mathbf{0},\mu_H=\mathbf{1}$) is
possible as long as $\mu_K \neq \mu_H$ elementwise (see \cref{lemma:units}).

In the theory that follows, we will find it convenient to restrict our focus to a bounded range of parameters. Namely, we
assume that, for some $M > 0$ and all groups $g$, the entries of $\beta_g$ and $\alpha_g$ lie in $[-M,M]$.
Note, however, that we do not restrict the choice of $M$, and the proportions we compute in all
cases end up having no dependence on $M$.

As the next two results demonstrate, under these choices, the fraction of parameters leading to sign flips is fairly large.
First, for the explained component, we find that, under our conditions, half of the parameter space $(\beta_H, \beta_K, \alpha_H, \alpha_K)$ corresponds to sign flips.
\begin{proposition} \label{prop:explained_fraction}
	Suppose $\mu_K = \mathbf{0}$ and $\mu_H = \mathbf{1}$.
	Choose any $M > 0$.
	Let $C_M \subset \R^{2d}$ be a cube with side length $2M > 0$ centered at the origin, and suppose $(\beta_H, \beta_K)$ jointly lie in this cube. Let $C_{flip} \subset C_M$ be the settings of $(\beta_H, \beta_K)$ that lead to a sign flip in the explained component.
	Then
	\begin{equation}
		\frac{\mathrm{Volume}(C_{flip})}{\mathrm{Volume}(C_M)} = \frac{1}{2}.
	\end{equation}
\end{proposition}
See \cref{proof:explained_fraction} for the proof.

Second, for the unexplained component, we find that the proportion of parameter space corresponding to sign flips nears 1/2 as the dimension of the covariate vector grows.
\begin{proposition} 
\label{prop:unexplained_fraction}
	Suppose $\mu_K = \mathbf{0}$ and $\mu_H = \mathbf{1}$.
	Let $C_M \subset \R^{2d+2}$ be a cube with side length $2M > 0$ centered at the origin, and suppose $(\beta_K, \beta_H, \alpha_K, \alpha_H)$ jointly lie in this cube.
	Let \( I_2 \sim \text{Irwin-Hall}(2) \) and \( J_{2d} \sim \text{Irwin-Hall}(2d) \) be independent random variables.
	The fraction of $C_M$ for which a sign flip occurs in the unexplained component is given by:
\begin{align} \label{eq:unexplained_fraction}
    P_d := \mathrm{Pr} \left( \{ I_2 > 1 \} \cap \{ J_{2d} < d + 1 - I_2 \}  \right) + \mathrm{Pr} \left( \{ I_2 < 1 \} \cap \{ J_{2d} > d + 1 - I_2 \} \right),
\end{align}
where the probability $\mathrm{Pr}$ is over the independent random variables $I_2, J_{2d}$.
Moreover, $P_d \underset{d \to \infty}{\longrightarrow} 1/2$.
\end{proposition}
See \cref{proof:unexplained_fraction} for the proof.
The joint Irwin-Hall probability in \cref{eq:unexplained_fraction} is hard to evaluate analytically, so we evaluate it numerically for various values of $d$; see \cref{fig:sign_flip_overview} (right).
For even modest dimension $d$, the fraction of parameter space leading to sign flips exceeds 40\%.

While our results above assume the standardization $\mu_K=\mathbf{0}, \mu_H=\mathbf{1}$, 
we confirm by simulations that we expect similar results when we take each coordinate of $\mu_g$ uniform
in the range $[-M,M]$ as well. See \cref{appendix:sign_flip_simulation} for details.

\subsection{Reconciling the theoretical and empirical prevalence of sign flips}
\label{sec:reconciling_theory_empirics}

Our main focus in \cref{sec:sign_flip_fraction} was establishing that a non-zero proportion of parameter
space yielded sign flips.
One limitation of this theory is that it assumes regression parameters are uniformly distributed, when in practice we may expect they follow non-uniform distributions.
Indeed, our earlier empirical exercise from \cref{tab:flip_summary_main}, finds that sign flips are less common than a direct application of theory would suggest. 
Even if we count total flips and ignore statistical significance, we see sign flips in 19\% of (i.e., 27 of 139) tested subsets for the OBD.
We now describe and interpret constraints on the parameter space that may produce fewer sign flips than implied by our theoretical results.

\textbf{Explained component.} First, consider the explained component. Under the standardization from \cref{sec:sign_flip_fraction},
we can simplify the characterization of when a sign flip happens.
\begin{remark} \label{rmk:no_explained_sign_flip}
	Suppose $\mu_K = \mathbf{0}$ and $\mu_H = \mathbf{1}$.
	Then there is a sign flip if and only if $\sign (\mathbf{1}^\top \beta_{H}) \neq \sign( \mathbf{1}^\top \beta_{K})$.
\end{remark}
This observation follows immediately from \cref{defn:sign_flip_explained}.
In some contexts, we may expect the elements of $\beta_g$ to be of the same sign in both groups.
Consider the ICU mortality example from \cref{sec:real_example}.
While we may expect risk factors to have different effect sizes for men and women a priori, there is little reason to think that many factors that increase men's mortality would decrease women's (or vice versa).
\cref{rmk:no_explained_sign_flip} could then explain why we observe explained component sign flips in only 11 out of the 139 subgroups we test (\cref{tab:flip_summary_explained}).

\textbf{Unexplained component.} Next, consider the unexplained component. In this case, a sign flip intuitively requires the intercept and slope effects to pull in opposite directions across groups. When both move in the same direction, so that the group with higher baseline outcomes also exhibits higher returns to covariates, the unexplained component preserves its sign regardless of the chosen reference group.
We now formalize these ideas.
\begin{assumption}[Aligned intercept--slope effects]
\label{ass:no_sign_flip}
$
    \sign(\Delta\alpha) = \sign(\mu_H^T \Delta\beta) = \sign(\mu_K^T \Delta\beta)
$
\end{assumption}

\begin{proposition}
    \label{prop:no_sign_flip}
    Under \cref{ass:no_sign_flip}, the unexplained components 
    \( U^{(H)}\) and  \(U^{(K)}\) share the same sign.
\end{proposition}
See \cref{proof:no_sign_flip} for the proof.
In some domains, intercept and slope advantages might often move in the same direction; 
\cref{prop:no_sign_flip} rules out sign flips in such cases.
To illustrate, consider an analysis of the gender wage gap that attributes the difference in average earnings between men and women to covariates (e.g., work experience) and to outcomes given covariates (e.g., how experience is remunerated).
\cref{ass:no_sign_flip} corresponds to the group with higher baseline wages ($\alpha_g$) also having higher returns to experience ($\beta_g$). 
\cref{prop:no_sign_flip} is one explanation for why sign flips are rare in our analysis of U.S. labor force data (\cref{appendix:us_labor_force}): \cref{ass:no_sign_flip} holds in 182 out of 400 comparisons we search (45.5\%), ruling out unexplained component sign flips in those cases.

However, \cref{ass:no_sign_flip} is not plausible in all domains.
In the ICU setting from \cref{sec:real_example}, the assumption requires the higher mortality population to also be more sensitive to observed covariates.
This relationship may fail to hold when important factors --- such as environmental exposure, chronic conditions, or family history of disease --- are not recorded in a researcher's dataset.
If such \textit{omitted covariates} are strong predictors of mortality, a population may have a high mortality rate that is less sensitive to \textit{observed covariates}.
Consistent with this intuition, \cref{ass:no_sign_flip} never holds across the 139 comparisons in the ICU data exercise.
In \cref{appendix:ovb}, we explore how the presence of omitted covariates can generate parameter differences $\Delta\beta$ and $\Delta\alpha$ that do preclude sign flips, without invoking \cref{ass:no_sign_flip}.


\section{Conclusion}
\label{sec:conclusion}

It is well known that the OBD depends on the choice of reference group. 
We here demonstrate that this asymmetry can result in substantively different conclusions between the two reference choices.
We present examples from real datasets as existence proofs. 
We show that sign reversals can occur in either the explained or unexplained component and are not solely products of model misspecification, small sample size, or an adversarial data generating process. 
As a practical conclusion, we recommend that data analysts using the (N)OBD always report results using both reference groups to accurately convey the degree of support behind substantive conclusions arising from the (N)OBD.

Our theoretical results show that the fraction of parameter space representing sign reversals is large for linear models and grows with the dimension of the feature space.
However, our empirical exercises reveal fewer sign reversals than our theoretical results suggest, implying that fitted linear regression parameters in real data analyses may not be well modeled by a uniform distribution over the parameter space. 
Future work can determine to what extent our theoretical findings hold for non-uniform parameter distributions, other measures of model complexity, and nonlinear models. 

\textbf{Broader impacts.}
The (N)OBD is widely used to draw conclusions and influence policy.
Our work shows that the choice of reference group in the (N)OBD can reverse substantive conclusions and characterizes conditions under which such reversals can occur.
We highlight the importance of reporting results for both references, even though this step is currently sometimes omitted in published papers
and government reports.
Thus, our work has the potential to improve decisions that arise from using the (N)OBD.
Our work does have one potential negative societal impact: that researchers learn they may be able to cherry-pick results through the choice of (N)OBD reference group.
However, a researcher with an agenda can also influence results through sample selection, feature engineering, and model choice; we think it unlikely that reference group choice is the only means to their end.
If one wishes to cherry-pick, there are already more than enough trees to harvest from, so to speak.
Further, we suspect (or, hope) that the majority of mis-use of the (N)OBD is from researchers unaware of its problems, and so hope that the net impact from our paper will be positive.


\section*{Acknowledgements}

DISTRIBUTION STATEMENT A. Approved for public release. Distribution is unlimited. This material is based upon work supported by the Combatant Commands under Air Force Contract No. FA8702-15-D-0001 or FA8702-25-D-B002. Any opinions, findings, conclusions or recommendations expressed in this material are those of the author(s) and do not necessarily reflect the views of the Combatant Commands.

Manuel Quintero was supported in part by the Fortunato and Catalina Brescia Fellowship at the Institute for Data, Systems, and Society (IDSS), Massachusetts Institute of Technology.
\bibliographystyle{icml2026}
\bibliography{OBD_signflip}


\appendix
\section{Extended literature review}
\label{appendix:lit_review_xAI}

In this appendix, we expand on the discussion in \cref{sec:intro} of why standard explainability methods do not address the hospital data scientist's question from \cref{sec:intro}, and of how the (N)OBD differs from these methods.
We then discuss how common modifications to the (N)OBD that might seem to remove the index number problem either do not or introduce new questions of arbitrariness similar to the index number problem.

\subsection{Explainability methods for attributing outcome behavior to covariates}
\label{appendix:xAI}

A common feature of modern explainability methods, including Shapley values \citep{Shapley1953, Lundberg2017, chen2025shapley}, Functional ANOVA \citep{stone1994, hooker2004discovering, hookerGeneralizedFunctionalANOVA2007, lengerich2020purifying, fumagalli2025fanova}, Accumulated Local Effects \citep{apley2020}, and partial dependence plots \citep{friedman2001greedy, liu2025trees}, is that they take as their input a fitted outcome model, or more generally a conditional mean function $x \mapsto \E[Y \mid X=x]$.
They ask how this function changes as one or more covariates change, or how much variation in the function can be attributed to particular (groups of) covariates.
This type of decomposition is well-suited to many tasks in machine learning explainability, but does not, by itself, separate differences between two populations into differences in the distribution of $X$ versus differences in the distribution of $Y \mid X$.
We elaborate on this distinction below for three common decomposition methods.

Functional ANOVA measures the importance of features in determining the output of a function $f$ and identifies underlying additive interactions between subsets of covariates \citep{hooker2004discovering, hookerGeneralizedFunctionalANOVA2007}.
It represents $f$ as a sum of an overall mean, main effects of individual covariates, pairwise interactions, tertiary interactions, etc. The definitions of these terms depends on a reference measure over the covariates that pins down the orthogonality and centering (mean zero) restrictions.
Accumulated Local Effects \citep[ALE;][]{apley2020} and partial dependence plots \citep[PDP;][]{friedman2001greedy} summarize how a fitted function changes with one covariate or a subset of covariates; for instance, how predicted mortality varies with admission heart rate, but they do not define a decomposition of the gap between two groups and at least it is not clear how to use them to do so; one can similarly define an additive decomposition of a function $f$ using PDP or ALE.
But Functional ANOVA, ALE, and PDP are tools for decomposing a \emph{single} function $f$, not for studying differences between two functions, which is the problem addressed by the (N)OBD.
We are aware of only one work applying ALE and Functional ANOVA to generalize the NOBD \citep{quintero2025}; however, \citet{quintero2025} shows this leads to serious issues beyond just the index number problem studied here.


Shapley-value-based methods are feature attribution tools that assign each covariate a portion of a prediction, or of the variation in predictions, by treating the covariates as players in a cooperative game whose payoff is the model's output \citep{Shapley1953, Lundberg2017, chen2025shapley}.
Again, though, Shapley values are useful for attributing the output of a single fitted model to features, not a situation in which there are two populations.
In \cref{sec:existing_index_number} we detail one way in which Shapley values could be used to produce a decomposition for differences between populations like the one produced by the (N)OBD, and this decomposition does not have an index number problem. However, we discuss in \cref{sec:existing_index_number} why such a possible solution actually introduces other issues.

\subsection{Adaptations of the Oaxaca--Blinder decomposition beyond linear models}
\label{appendix:obd_generalized}

The OBD was developed for linear models \citep{oaxaca1973, blinder1973} and remains widely used in this form across applied health, labor, and policy research \citep{paradaContzen2025, barhaim2023, cartwright2021}.
Still, recent work has paired the OBD with flexible ML estimators to relax the linearity assumption while retaining the between-population counterfactual structure of \cref{eq:OB_H_Counterfactual,eq:OB_K_Counterfactual}.
We review these efforts here.
\citet{bach2024heterogeneity} use double-lasso to fit a high-dimensional wage equation and quantify how the US gender wage gap varies with characteristics such as marital status, race, occupation, and education.
\citet{mao2025double} apply the double machine learning framework to the Oaxaca--Blinder decomposition, deriving $\sqrt{n}$-consistent estimators of the composition (i.e., explained) and structure (i.e., unexplained) effects under high-dimensional mixed covariates via Neyman orthogonal scores.
\citet{quintas2024multiply} introduce a multiply-robust estimator for causal change attribution that combines regression and reweighting, remains consistent under partial misspecification of the nuisance ML models, and is asymptotically normal.
\citet{flachaire2025decomposing} reformulate the decomposition in terms of potential outcomes and pair the weighted reference outcome of \citet{neumark1988} with double machine learning to avoid the common-support and trimming requirements that arise when the reference is taken to be one of the two groups.
Closer to our setting, \citet{quintero2025} examine what happens when off-the-shelf functional decompositions, specifically Functional ANOVA and ALE, are plugged into the OBD framework, and show that they can misattribute differences in $X$ to differences in $Y \mid X$.

In all of these methods, $\E_H[Y \mid X]$ and $\E_K[Y \mid X]$ are estimated by flexible $\hat{M}_H(X)$ and $\hat{M}_K(X)$, and counterfactual means are constructed as in \cref{eq:NOBD_H_estimated}.
Such methods retain the between-population target of the OBD and so address the limitation of the explainability methods above.
However, with the exception of \citet{flachaire2025decomposing}, who fix a single Neumark-style weighted reference \citep{neumark1988}, this line of work does not engage with the issue we study.
That is, even once linearity is relaxed, the (N)OBD still requires choosing one of the two groups as a reference, and this choice can affect the numerical result; more importantly, as we show, it can reverse substantive conclusions.
Improving the flexibility of $\hat{M}_g$ thus does not remove the index number problem; if anything, richer models may make the issue worse, as \cref{tab:flip_summary_main} shows for our main empirical exercise.

In short, adaptations of the OBD and explainability methods address different parts of the overall problem: explainability methods describe how a single fitted function depends on covariates, while the NOBD improves the estimation of the conditional mean functions used to form counterfactual means.
Our paper studies distinct questions --- whether the two natural (N)OBD references can yield substantively different conclusions, and how common this sensitivity is.

\subsection{Existing solutions to the index number problem}
\label{sec:existing_index_number}
A few proposed methods exist that compute decompositions similar to the (N)OBD, and yet do not have the index number problem.
We describe the approach of these options here and discuss how they do not resolve the issue of arbitrariness discussed in this paper.

\textbf{Shapley Values.} Shapley values \citep{Shapley1953} are widely used in the machine learning literature to understand a models' predictions.
A straightforward use of them would be to estimate ``explained'' or ``unexplained'' contributions to $\Delta_Y$.
This follows from the standard Shapley construction, which treats each reference-based decomposition as a ``player'' and assigns contributions by averaging marginal effects over all orderings; with two ``players'' (the two reference directions $H$ and $K$), there are two equally weighted permutations, $(H,K)$ and $(K,H)$, so the Shapley value reduces to the simple average, for the explained component, $E = \tfrac{1}{2}E^{(H)} + \tfrac{1}{2}E^{(K)}$.

However, our paper complicates the interpretation of $E^{(H)}$ and $E^{(K)}$ (and likewise for the unexplained components). 
If the explained component can take on opposite signs depending on the reference group, it is ambiguous whether the distributions of covariates widen or shrink the gap in outcomes between $H$ and $K$.
It is not clear that averaging across the reference groups appropriately resolves this ambiguity; for instance, if $E^{(H)}$ and $E^{(K)}$ can take opposite signs, averaging them destroys information about their magnitude.
For example, $\tfrac{1}{2}E^{(H)} + \tfrac{1}{2}E^{(K)}$ can take on small positive values if $E^{(H)}$ and $E^{(K)}$ are both small and positive, or both large but and of opposing signs. 
A researcher may wish to distinguish between these scenarios, but the Shapley proposal above does not.

\textbf{Pooled approaches.}
Another approach to eliminate the index number problem is to define some ``pooled'' group and then measure deviations from this group in a three-term decomposition. For example, in the case of the (linear) OBD, one defines a pooled $\beta^*$ and then computes:
\begin{equation}
  \E_H[Y] - \E_K[Y] = (\mu_H - \mu_K)^T \beta^* + \mu_H^T (\beta_H - \beta^*) + \mu_K^T(\beta^* - \beta_K).
\end{equation}
The first term is interpreted as differences due differences in covariates, the second differences due to differences in $H$'s $Y \mid X$ from the baseline, and the third differences due to differences in $K$'s $Y \mid X$ from the baseline.

The question that remains, then is how to define $\beta^*$.
A typical choice is to use some sort of ``average'' of the two groups.
One approach from \citet{oaxacaRansom1998}, defines some $w \in [0,1]$ and sets $\beta^* = (1-w)\beta_K + w\beta_H$.
A second approach, from \citet{neumark1988}, is to create a pooled dataset consisting of all samples from $H$ and $K$, and estimate $\beta^*$ from this dataset.
However, these approaches suffer from issues of arbitrariness as well.
First, how does one pick $w$ in the Oaxaca-Ransom approach?
Second, why is defining $\beta^*$ one way more reasonable than the other?
If we assume both are equally reasonable, then, given our results for the OBD, we feel that one should be highly suspicious of substantive conclusions being robust to the choice of $\beta^*$.
Finally, justifications for a choice of $w$ or $\beta^*$ in one research setting may not apply in others.
For example, in analyses of the labor market, $\beta^*$ is sometimes intended to represent the $Y \mid X$ relationship in the absence of wage discrimination \citep{neumark1988}. 
There may be no corresponding concept to serve as a reference point in other research contexts; or, there may be no feasible approach for estimating such a reference point.

\textbf{Intermediate covariate distributions.}
\citet{dinardoFortinLemieux1996} propose a variant of the OBD with a more fine-grained decomposition of the explained component, and \citet{quintero2025} propose a NOBD with a similarly fine-grained explained decomposition (\citet{quintero2025} also decomposes the unexplained component).
Both cases define a number of intermediate distributions over covariates $P_1, \dots, P_I$, where $P_1 = K_X$ and $P_I = H_X$. E.g., in the case of \citet{quintero2025}, $P_2$ is equal to $K$, except that the marginal of the first covariate is equal to the marginal of $H_X$.
In either case, for each $i = 2, \dots, I$, intermediate explained components are defined as one of:
\begin{align}
  \E_{X \sim P_i}[ \E_{Y \sim H}[Y \mid X] ] - \E_{X \sim P_{i-1}} [ \E_{Y \sim H} [Y \mid X]] \\
  \E_{X \sim P_i}[ \E_{Y \sim K}[Y \mid X] ] - \E_{X \sim P_{i-1}} [ \E_{Y \sim K} [Y \mid X]].
\end{align}  
There are two issues with such approaches.
First, \citet{dinardoFortinLemieux1996} does not change how the the unexplained component is computed over the OBD, so the same exact index number problems with the unexplained component detailed here apply to the decomposition of \citet{dinardoFortinLemieux1996}.
Second, one now has to justify not just a choice of $H$ or $K$ as a reference group, but instead a whole ordering of the $P_i$.
This problem is explicitly called out by \citet{quintero2025}, who choose to fix an arbitrary order.
Given our results here, we suspect that conclusions may not be robust to the ordering of the $P_i$, and we leave such investigation for future work.

\section{Details on Real-Data Examples}
\subsection{ICU Mortality Example}
\label{appendix:icu_details}

This appendix provides full details on how we constructed the ICU data subsets, prepared the analytic dataset, estimated the OBD, identified sign reversals, and computed bootstrap standard errors for all data subsets used in \cref{sec:real_example}. The analyses use the PhysioNet ICU cohort collected for the study of in-hospital mortality \citep{goldberger2000physiobank}. The raw cohort contains routine admission measurements for $3,997$ patients with known gender ($2,246$ women and $1,751$ men). We merge these records with the corresponding mortality outcomes and restrict attention to patients with complete information on the covariates used in the decomposition: age, ICU type, heart rate (HR), mean arterial pressure (MAP), temperature, urine output, and the indicator for in-hospital death. After removing observations with missing covariates, the final analytic dataset contains $3,429$ patients. All data subset definitions, model specifications, and statistical procedures described below are applied to this cleaned dataset.

\subsection*{Model specification}

For each data subset, we fit separate group-specific models of in-hospital mortality for men and for women. Let $g \in \{\text{men}, \text{women}\}$, let $Y$ be the indicator for in-hospital mortality, and let $X$ be the vector of admission covariates (heart rate, mean arterial pressure, temperature, urine output, age, and ICU type). For each group $g$, we fit an estimator $\hat{M}_g(X)$ of the conditional expectation $\E[Y \mid X, g]$ and use \cref{eq:OB_H_Counterfactual,eq:OB_K_Counterfactual}; see \cref{sec:ob_real}. We consider five model classes, specified as follows.

\textbf{Linear regression.} We take
\[
    \E[Y \mid X, g] \;=\; \alpha_{g} + X^{\top}\beta_{g},
\]
with the group-specific intercept $\alpha_{g}$ and slope $\beta_{g}$ estimated by ordinary least squares. Although $Y$ is binary, this linear probability model is the object to which the classical OBD applies and remains standard in applied work \citep{jann2008, edokaChangesCatastrophicHealth2017, sujin_LPM_OB, mweembaGapSelfRatedHealth2023}.

\textbf{Logistic regression.} Within each group, we standardize the covariates to zero mean and unit variance and fit a logistic regression of $Y$ on the standardized covariates with the default $L_2$ penalty.

\textbf{Neural network.} Within each group, we standardize the covariates and fit a two-hidden-layer fully-connected ReLU network with hidden widths $(32, 16)$, $L_2$ weight decay $10^{-4}$, the Adam optimizer \citep{kingma2015adam}, and max\_iter $=2000$.

\textbf{XGBoost.} We fit a gradient-boosted tree classifier with $250$ estimators, maximum depth $3$, learning rate $0.05$, row subsampling rate $0.9$, column subsampling rate $0.9$, $L_2$ leaf regularization $\lambda = 1.0$, histogram tree method, and the logistic (log-loss) objective.

\textbf{TabPFN.} We use a modern transformer-based tabular model, TabPFN \citep{hollmann2023tabpfn}\footnote{TabPFN uses a custom license: \url{https://huggingface.co/Prior-Labs/tabpfn_2_6/blob/main/LICENSE}}, a pre-trained in-context classifier that maps a training set directly to a predictive distribution for $Y \mid X$. We use the publicly released v2.6 weights and use the v7.1.1 TabPFN code for inference. We run the model locally on a single NVIDIA A100 GPU with $80$~GB of vRAM.

For all five models, we take the predicted probability of in-hospital death as $\hat{M}_g(X)$ and estimate each NOBD counterfactual mean $\E_{g'}[\E_g[Y \mid X]]$ by the within-sample average $\frac{1}{|g'|}\sum_{i \in g'} \hat{M}_g(X_i)$.

\subsection*{Data subset construction}
The goal of this analysis was not to target any particular physiological pattern, but rather to explore systematically whether sign reversals arise in realistic data partitions that a practitioner might naturally encounter.

We constructed $139$ data subsets in total. These sets include: quartiles of each continuous admission vital (heart rate, MAP, temperature, urine output);
deciles of age; ICU-type specific cohorts; clinically motivated threshold subsets (HR $>100$, MAP $<65$, Temp $> 38^\circ$C, Urine $>1000$ mL); fifty random $50\%$ subsamples of the cohort; fifty random $30\%$ subsamples of the cohort.

For each data subset and each of the five models in our model class, we compute the decomposition under both reference groups, taking men and women in turn. We say that the explained or unexplained component sign flips in that data subset if its estimated value has opposite signs under the two reference choices. To assess whether such flips are robust to sampling variability, we next describe the bootstrap procedure we use to compute standard errors and $p$-values for each component.

\textbf{Bootstrap.}
For each data subset and each model, we compute bootstrap standard errors and $p$-values for the OBD components. The bootstrap procedure follows the standard Oaxaca--Blinder framework: men and women are resampled within group with replacement, preserving group sizes. For each bootstrap replication $b=1,\dots,B$, we re-fit the corresponding group-specific model and re-estimate the OBD components.

Let $\hat\theta$ denote any Oaxaca--Blinder component (for example, the explained component under the women reference model), and let $\hat\theta^{(1)},\dots,\hat\theta^{(B)}$ denote its bootstrap replicates. The bootstrap standard error and mean are
\[
    \widehat{\operatorname{se}}(\hat\theta)
    =
    \sqrt{
    \frac{1}{B-1}\sum_{b=1}^B
    \bigl(\hat\theta^{(b)} - \bar\theta\bigr)^2
    },
    \qquad
    \bar\theta = \frac{1}{B}\sum_{b=1}^B \hat\theta^{(b)}.
\]
$P$-values are computed using a normal approximation,
\[
    p = 2\left(1 - \Phi\left(\frac{|\hat\theta|}{\widehat{\operatorname{se}}(\hat\theta)}\right)\right).
\]

We say that a sign flip is significant at level $\alpha$ if, in addition to having opposite signs across reference groups, at least one of the two reference-specific estimates has $p < \alpha$. We summarize the total number of data subsets exhibiting any sign flip under each model in \cref{tab:flip_summary_main} in \cref{sec:real_example}. Here, \cref{tab:flip_summary_explained,tab:flip_summary_unexplained} below report the same counts separately for the explained and unexplained components.

We observed that empirically, the count of data subsets exhibiting a sign flip grows with model complexity for this particular data analysis. With the linear and logistic models showing the least amount of sign flips and only a couple remain significant at conventional thresholds. While the neural network and XGBoost show a substantial fraction of data subsets exhibits significant flips even at the $5\%$ level. Importantly, many of the flipped data subsets arise within the random subsamples, indicating that the OBD can yield unstable conclusions even when applied to simple resampled partitions of the same population. Several physiologically interpretable groups also displayed sign changes. The HR quartile two data subset examined in \cref{sec:real_example} is one such case, and serves as a concrete illustration of how these reversals can appear in clinically meaningful settings.

\begin{table}[!ht]
    \caption{Number of data subsets (out of $139$) exhibiting a sign flip in the explained component, by model. Columns $10\%$, $5\%$, and $1\%$ report the number of flipped data subsets that are significant at the corresponding level.}
    \label{tab:flip_summary_explained}
    \centering
    \small
    \begin{tabular}{lcccc}
    \toprule
    Model & Total flips & $10\%$ & $5\%$ & $1\%$ \\
    \midrule
    Linear      & 11 &  2 &  2 &  2 \\
    Logistic    &  9 &  2 &  2 &  1 \\
    Neural net  & 42 & 26 & 17 &  6 \\
    XGBoost     & 33 & 17 & 10 & 3 \\
    TabPFN      & 33 & 21 & 19 & 5 \\
    \bottomrule
    \end{tabular}
\end{table}

\begin{table}[!ht]
    \caption{Number of data subsets (out of $139$) exhibiting a sign flip in the unexplained component, by model. Columns $10\%$, $5\%$, and $1\%$ report the number of flipped data subsets that are significant at the corresponding level.}
    \label{tab:flip_summary_unexplained}
    \centering
    \small
    \begin{tabular}{lcccc}
    \toprule
    Model & Total flips & $10\%$ & $5\%$ & $1\%$ \\
    \midrule
    Linear      & 17 &  0 &  0 &  0 \\
    Logistic    & 21 &  0 &  0 &  0 \\
    Neural net  & 49 & 10 &  6 &  1 \\
    XGBoost     & 47 & 3 &  1 &  0 \\
    TabPFN      & 56 & 2 &  0 &  0 \\
    \bottomrule
    \end{tabular}
\end{table}

A substantial number of these sign reversals are significant at conventional thresholds, indicating that the instability documented in \cref{sec:real_example} is an inherent behavior of the OBD in this clinical dataset.

\subsection*{HR quartile 2: results across models}
\label{appendix:hrq2}

\cref{sec:real_example} focuses on the HR quartile two data subset under the OBD. In \cref{tab:hrq2_bootstrap_models}, we report the corresponding decomposition under each of the five models in our model class. The linear, logistic, neural-network, and XGBoost results are obtained from the bootstrap procedure described above with $B = 1{,}000$ replications.

\begin{table}[!ht]
    \caption{OBD of the gender mortality gap in HR quartile two ($n = 816$), across all five models. The total gap is men's mean minus women's mean. Standard errors in parentheses are computed via bootstrap resampling with $1{,}000$ replications. $p$-values are based on a normal approximation (Wald test). Significance levels: * $p<0.10$, ** $p<0.05$, *** $p<0.01$.}
    \label{tab:hrq2_bootstrap_models}
    \centering
    \scriptsize
    \begin{tabular}{@{}lc*{4}{c}@{}}
    \toprule
    \multirow{2}{*}{Model} &
    \multirow{2}{*}{Gap (M--F)} &
    \multicolumn{2}{c}{Women reference} &
    \multicolumn{2}{c}{Men reference} \\
    \cmidrule(lr){3-4} \cmidrule(lr){5-6}
    & & Explained & Unexplained & Explained & Unexplained \\
    \midrule

    Linear
        & \makecell{$0.035$ \\ $(0.023)$}
        & \makecell{$0.021^{***}$ \\ $(0.007)$}
        & \makecell{$0.014$ \\ $(0.025)$}
        & \makecell{$-0.007$ \\ $(0.010)$}
        & \makecell{$0.042^{*}$ \\ $(0.025)$} \\[0.3em]

    Logistic
        & \makecell{$0.035$ \\ $(0.023)$}
        & \makecell{$0.026^{***}$ \\ $(0.009)$}
        & \makecell{$0.009$ \\ $(0.025)$}
        & \makecell{$-0.006$ \\ $(0.008)$}
        & \makecell{$0.041^{*}$ \\ $(0.024)$} \\[0.3em]

    Neural net
        & \makecell{$0.035$ \\ $(0.023)$}
        & \makecell{$0.001$ \\ $(0.024)$}
        & \makecell{$0.033$ \\ $(0.031)$}
        & \makecell{$0.053^{**}$ \\ $(0.022)$}
        & \makecell{$-0.018$ \\ $(0.025)$} \\[0.3em]

    XGBoost
        & \makecell{$0.035$ \\ $(0.023)$}
        & \makecell{$0.004$ \\ $(0.016)$}
        & \makecell{$0.031$ \\ $(0.025)$}
        & \makecell{$0.029^{**}$ \\ $(0.014)$}
        & \makecell{$0.006$ \\ $(0.021)$} \\[0.3em]

    TabPFN
        & \makecell{$0.035$ \\ $(0.023)$}
        & \makecell{$0.012$ \\ $(0.016$}                                                                                                                                                                                                                                                                              
        & \makecell{$0.023$ \\ $(0.026)$}
        & \makecell{$0.023^{*}$ \\ $(0.013)$} 
        & \makecell{$0.012$ \\ $(0.021)$} \\[0.3em]
    \bottomrule
    \end{tabular}
\end{table}

The linear and logistic OBD exhibit a consistent explained-component sign flip that is significant at the $1\%$ level, with very similar magnitudes: in both cases, the explained component is significantly positive under the women reference and flips to a small negative (non-significant) value under the men reference. The neural network and XGBoost both point to the men reference estimate being significantly positive at the $5\%$ level.

\cref{fig:hr_quartile2_hist} presents the histogram for the $816$ patients in HR quartile two, showing the distribution of mortality outcomes used in the linear probability models displayed in \cref{fig:sign_flip_overview} (left). \cref{tab:delta_mu_betas} shows the covariate mean differences and group-specific coefficients for HR quartile two that underlie the sign reversal documented in \cref{sec:real_example}.

\begin{figure}[!ht]
    \centerline{\includegraphics[width=0.5\linewidth]{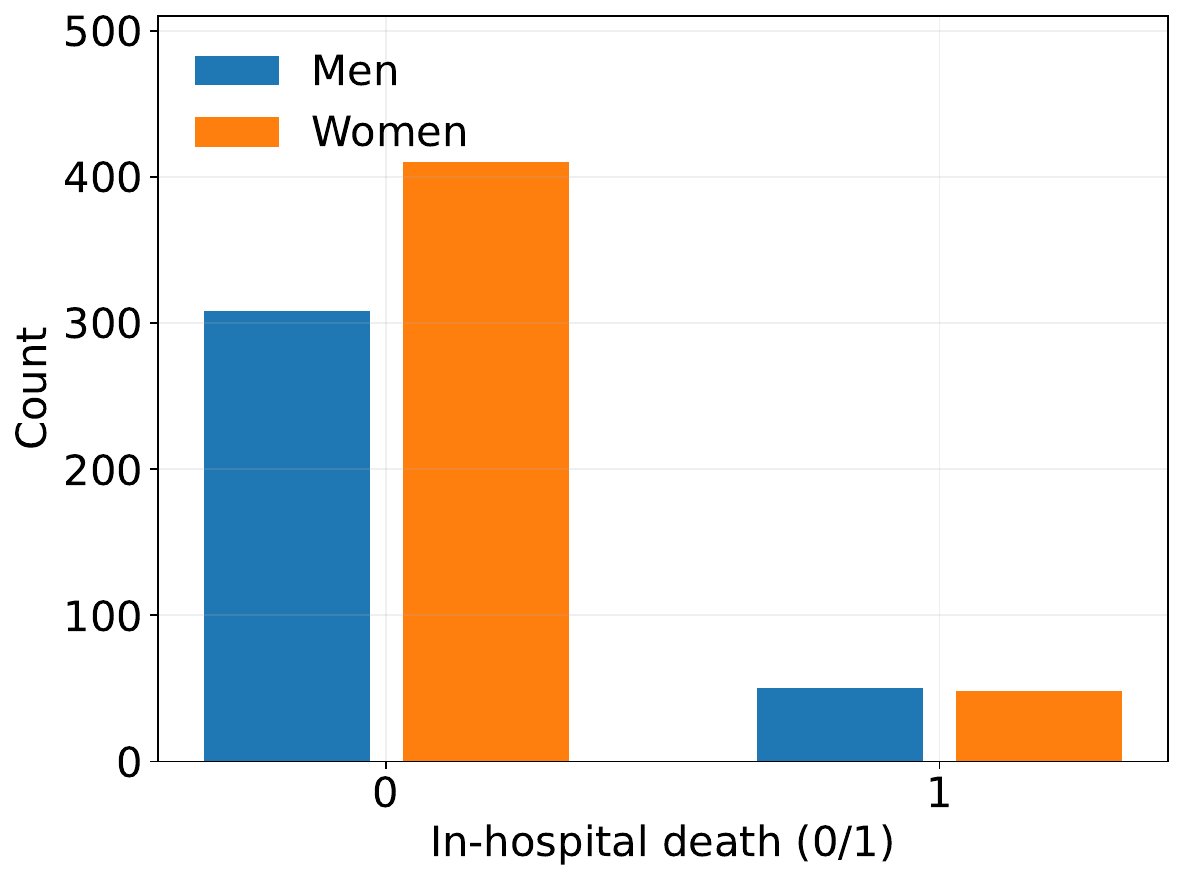}}
    \caption{Mortality outcome histogram for HR quartile two data subset.}
    \label{fig:hr_quartile2_hist}
\end{figure}

\begin{table}[!ht]
    \caption{Covariate mean differences and group-specific coefficients in HR quartile two. $\Delta\mu$ denotes the difference in covariate means between men and women within HR quartile two (men minus women). $\beta_{\text{Men}}$ and $\beta_{\text{Women}}$ are coefficients from group-specific linear probability models for in-hospital mortality.}
    \label{tab:delta_mu_betas}
    \centering
    \small
    \begin{tabular}{lrrr}
    \toprule
    Covariate 
    & $\Delta\mu$ 
    & $\beta_{\text{Men}}$ 
    & $\beta_{\text{Women}}$ \\
    \midrule
    Intercept & $0.0000$   & $1.1813$  & $-0.8515$ \\
    Age       & $3.8938$   & $-0.0000$ & $0.0039$ \\
    ICU Type  & $0.0095$   & $0.0173$  & $0.0116$ \\
    HR        & $0.4982$   & $-0.0026$ & $0.0056$ \\
    MAP       & $-0.6669$  & $0.0004$  & $0.0010$ \\
    Temp      & $0.3604$   & $-0.0230$ & $0.0028$ \\
    Urine     & $-39.8976$ & $-0.0001$ & $-0.0001$ \\
    \bottomrule
    \end{tabular}
\end{table}
    
\subsection{U.S.\ Labor Force Example}
\label{appendix:us_labor_force}
This appendix provides full details on how we constructed the U.S.\ labor force subgroups, prepared the analytic dataset, estimated the (N)OBD, and identified sign reversals for all subgroups. 
The raw data come from the 2016 American Community Survey 1-Year Public Use Microdata Sample, which contains 1,625,282 observations.
We apply the same sample construction as \citet{bach2024heterogeneity}, who study the gender wage gap.
We restrict the sample to adults between the ages of 25 and 65, who were employed and at work as civilians, worked at least 35 hours per week, at least 50 weeks per year, and with both earnings and income above \$12,687.50 (which corresponds conservatively to an hourly wage of \$7.25).
The resulting sample contains 871,849 observations.

We purposely construct a large set of analyses to search for sign flips in different subsets of the data, comprising all combinations of the factors in \cref{tab:us_census_design}.

\begin{table}[!ht]
    \caption{Settings for sign flip search in the U.S.\ labor force dataset.}
    \label{tab:us_census_design}
    \centering
    \footnotesize
    \setlength{\tabcolsep}{4pt}
    \begin{tabular}{@{}ll@{}}
    \toprule
    Subset \( S \) & Group (\( H, K \)) \\
    \midrule
    26 industries (NAICS 2) & Race (black, white)                  \\
    50 states + D.C,        & Sex (male, female)                   \\
                            & Nativity (native born, foreign born) \\
    \bottomrule
    \end{tabular}
\end{table}

There are \( 77 \times 3 = 231 \) combinations of subsets and groups, each corresponding to a particular subgroup an analyst could perform the (N)OBD within.
We consider two outcomes: the natural logarithm of income (the outcome from \citet{bach2024heterogeneity}), and a binary indicator for having health insurance.
We exclude 11 cases where there are fewer than 50 observations in either $H$ or $K$ (), or $|\Delta Y| < 50$; this intended to focus more on analyses where the data may be informative and differences in outcomes may be practically meaningful.
The result is 217 subgroups for log-income, and 183 subgroups for health insurance status.
For each analysis, we assemble a candidate feature set from state of residence, industry of occupation (NAICS 2 code), sex, race (black, white, or other), nativity, completion of a Bachelor's degree, and marital status.
We exclude features that define the particular subset and group at hand, to produce a feature set for model fitting.
For example, when examining male-female differences in the state of Wyoming, we remove sex and state from the feature set.

We then perform the (N)OBD by estimating $\E[Y \mid X, g, s]$ with the specified model, 
where \( Y \) and \( X \) are the outcome and covariates, 
\( g \) indexes the group (i.e. male or female sex);
and \( s \) indicates a subset of the data (i.e., the state of Utah).
The models are:
\begin{itemize}
    \item Ordinary least squares regression.
    \item Logistic regression (when $Y$ is binary), unpenalized.
    \item lightgbm (gradient-boosted trees). We set a learning rate of 0.25, number of leaves to 7, and select the number of trees by 3-fold cross-validation with early stopping at 10 rounds. All other parameters are set to default values.
\end{itemize}
We attempted to implement neural networks for regression (both outcomes) and classification (health insurance status only), but faced repeated software errors while optimizing model weights.
We use the same bootstrap inference procedure as in \cref{appendix:icu_details} for ordinary least squares and logistic regression. 
We do not perform inference for lightgbm due to computational costs of bootstrapping; future work can assess the statistical significance of the results with appropriate resources.
All computation was performed on an Apple M2 Max chip with 32GB ram.

\cref{tab:flips_census_income,tab:flips_census_insurance} report the number of sign flips for the log income and health insurance status outcomes, respectively.
Sign flips occur in both the explained and unexplained components for all tested model types.

\begin{table}[!ht]
    \caption{Number of subgroups (out of $217$) exhibiting a sign flip in the explained and unexplained component, for the earnings outcome. Columns $10\%$, $5\%$, and $1\%$ report the number of flipped data subsets that are significant at the corresponding level. No significance testing was performed for lightgbm due to computational costs of bootstrap inference.}
    \label{tab:flips_census_income}
    \centering
    \small
    \begin{tabular}{l*{12}{c}}
    \toprule
          & \multicolumn{4}{c}{Either Component} & \multicolumn{4}{c}{Explained Component} & \multicolumn{4}{c}{Unxplained Component} \\
          \cmidrule(lr){2-5} \cmidrule(lr){6-9} \cmidrule(lr){10-13}
    Model       & Total & $10\%$ & $5\%$ & $1\%$ & Total & $10\%$ & $5\%$ & $1\%$ & Total & $10\%$ & $5\%$ & $1\%$ \\
    \midrule
    Linear      & 35 & 14 & 11 & 7 & 24 & 10 & 8 & 5 & 11 & 4 & 3 & 2 \\
    lightgbm    & 29 & -  & -  & - & 23 & -  & - & - & 6  & - & - & - \\
    \bottomrule
    \end{tabular}
\end{table}

\begin{table}[!ht]
    \caption{Number of subgroups (out of $183$) exhibiting a sign flip in the explained and unexplained component, for the (binary) health insurance status outcome. Columns $10\%$, $5\%$, and $1\%$ report the number of flipped data subsets that are significant at the corresponding level. No significance testing was performed for lightgbm due to computational costs of bootstrap inference.}
    \label{tab:flips_census_insurance}
    \centering
    \small
    \begin{tabular}{l*{12}{c}}
    \toprule
          & \multicolumn{4}{c}{Either Component} & \multicolumn{4}{c}{Explained Component} & \multicolumn{4}{c}{Unxplained Component} \\
          \cmidrule(lr){2-5} \cmidrule(lr){6-9} \cmidrule(lr){10-13}
    Model       & Total & $10\%$ & $5\%$ & $1\%$ & Total & $10\%$ & $5\%$ & $1\%$ & Total & $10\%$ & $5\%$ & $1\%$ \\
    \midrule
    Linear      & 60 & 35 & 33 & 24 & 44 & 33 & 31 & 23 & 17 & 2 & 2 & 1 \\
    Logistic    & 77 & 23 & 16 & 10 & 60 & 23 & 16 & 10 & 25 & 1 & 0 & 0 \\
    lightgbm    & 97 & -  & -  & -  & 23 & -  & - & - & 6  & - & - & - \\
    \bottomrule
    \end{tabular}
\end{table}

\section{A simulated example from healthcare}
\label{sec:simulated_example}
The results in \cref{sec:real_example} may raise the question of whether sign reversals are always consequences of model misspecification or sampling variability? 
\cref{sec:signflips} provide theoretical arguments that sign reversals are not, in fact, purely due to misspecification or noise.
We complement those arguments here with a simulated example that exhibits a sign reversal under correct model specification and in population-level parameters.


\textbf{Motivation.}
While our example is purely simulated, it is motivated by the following real-life concerns and observations.
Body mass index (BMI) is a widely recognized and routinely measured marker of wellness \cite{heymsfield2016} and is strongly linked to cardiometabolic outcomes, such as elevated blood pressure \cite{brown2000body}.
A researcher studying cardiac health may use commonly observed covariates like BMI to investigate whether differences in blood pressure between two groups are explainable by differences in basic markers of health between the groups, or differences in how these markers map to blood pressure.

Our example imagines two groups of adults in primary care being followed for elevated blood pressure, using systolic blood pressure (SBP) at a follow-up visit as the outcome $Y$. The BMI ($X$) of each adult is measured at an intake visit.
The first, Group $H$, comes from a higher-resource setting, such as an urban neighborhood with structured care pathways, frequent follow up, and better access to healthy food options. 
The second, Group $K$, comes from a lower-resource setting, where patients face barriers such as longer travel to clinics, infrequent follow up, and limited access to nutritious food.
These factors shape both the composition of each group (including average BMI) and the structure of the SBP--BMI relationship. We follow a clinical literature finding an approximately linear relationship between BMI and SBP \cite{kaufman1997, cappuccio2008, chen2023}; in particular, we assume a separate linear relationship in each group.

While an OBD cannot establish causality without further assumptions, it can suggest where one might look for potential drivers of observed inter-group disparities.
If the explained component is effectively zero and the unexplained component (differences in $Y|X$) dominates (with a sign in the direction of the observed outcome difference), it may highlight the need to investigate differences in medical care, such as adherence support or access to antihypertensives.
If the unexplained component is effectively zero and the explained component (differences in $X$) dominates (with a sign in the direction of the observed outcome difference), it may be most impactful to target interventions on weight management, nutrition, and physical activity. 

\textbf{Simulated data.}
We generate the BMI ($X$, in kg/m$^2$) of an individual in group \( g \) from a group-specific normal distribution.
We simulate \( N = 5000 \) individuals, divided across two groups. Body mass index (BMI) is drawn from normal distributions reflecting realistic population profiles:
\begin{align*}
    \text{BMI}_g &\sim \mathcal{N}(\mu_g ,\, \sigma_g^2) \quad \text{kg/m}^2,
\end{align*}
where \( \mu_H = 25 \), \( \mu_K = 27 \), and \( \sigma_H = \sigma_K = 4 \). Within each group, systolic blood pressure is modeled as a noisy, linear function of BMI:
\begin{align}
    \text{SBP}_g &= \alpha_g + \beta_g \times \text{BMI}_g + \varepsilon_g.
\end{align}
Here $\alpha_g \in \mathbb{R}$ is the group-specific intercept (in units of mmHg); $\beta_g \in \mathbb{R}$ is the group-specific slope, representing the average change in SBP (mmHg) per unit increase in BMI (kg/m$^2$); and $\varepsilon$ reflects individual variability in SBP that is uncorrelated with BMI. We take $\varepsilon$ independent and identically distributed across adult patients, with the same (mean-zero Gaussian) distribution in each group.

The noise term \( \varepsilon_g \) captures unobserved individual-level variation within each group, such as daily measurement fluctuations, dietary salt intake, or transient stress. 
Group~H (higher-resource settings) has a lower mean BMI, consistent with greater access to healthy food, safe environments for physical activity, and preventive care. 
In contrast, Group~K (lower-resource or marginalized settings) has a mean BMI two units higher, reflecting structural disparities related to nutrition, transportation, and economic constraints. 
The shared standard deviation of 4~kg/m\(^2\) reflects typical within-group heterogeneity in adult primary care populations \citep{block2013population}.

We assume \( \beta_K > \beta_H \), reflecting that in Group~K, each additional unit of BMI increases SBP more steeply, due to slower titration, less frequent follow up, greater adherence variability, and fewer lifestyle resources.
In contrast, the flatter slope in Group~H reflects partial buffering through timely medication adjustments, consistent coverage, and access to lifestyle support.
The intercepts \( \alpha_H \) and \( \alpha_K \) are chosen so that the regression lines intersect near the midpoint of the BMI range, ensuring that mean SBP differences are primarily driven by BMI distribution and slope disparity.
\cref{fig:bmi_sbp} shows the fitted linear models to the synthetic data.

\begin{table}[h]
\caption{Key chosen and derived parameters by group. BMI and SBP are in kg/m$^2$ and mmHg respectively. $\alpha$ is in mmHg; $\beta$ is in mmHg per kg/m$^2$.}
\label{tab:dgp_short}
\begin{center}
\small
\begin{tabular}{lcccc}
\toprule
Group & Mean BMI & $\alpha$ & $\beta$ & Mean SBP \\
\midrule
$H$ & 25.0 & 110.4 & 1.0 & 135.4 \\
$K$ & 27.0 & 100.0 & 1.4 & 137.8 \\
\bottomrule
\end{tabular}
\end{center}
\vskip -0.1in
\end{table}

\cref{fig:bmi_sbp} shows the population SBP--BMI relationships for both groups (two lines). The figure also shows data samples
to illustrate the BMI distribution in each group --- as well as how the intercepts, slopes, and BMI distributions interact.
Overall, Group $H$ has lower average SBP than Group $K$, owing both to lower average BMI and a shallower slope relating BMI to SBP, perhaps reflecting improved access and adherence to treatment, or availability of non-pharmacological support.
However, Group $H$ has a higher baseline SBP ($\alpha_H$), perhaps reflecting genetic factors in Group $H$.

\begin{figure}[h]
\centerline{\includegraphics[width=0.5\linewidth]{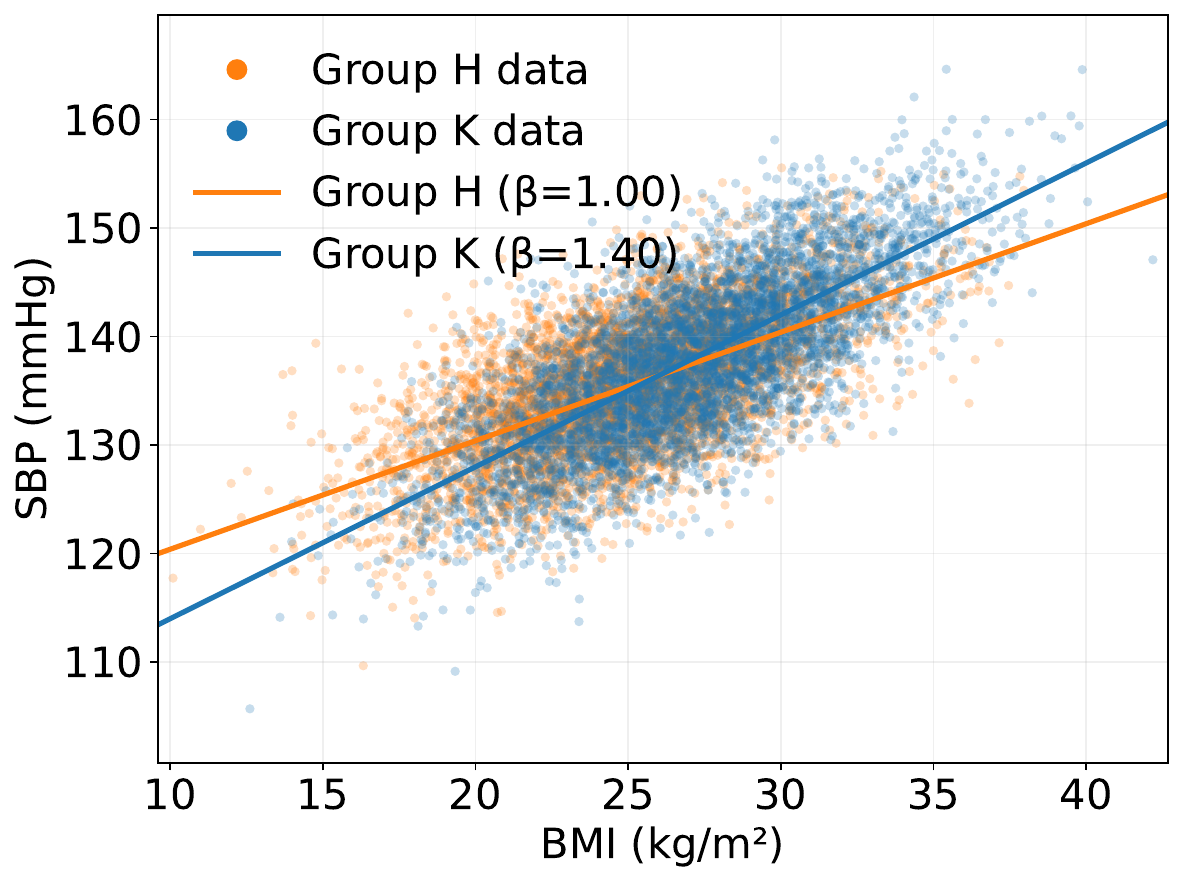}}
\caption{Population linear models of SBP on BMI by group and data samples.}
\label{fig:bmi_sbp}
\end{figure}

\textbf{Analysis.} Since we have access to all population quantities in this example, there is no need to fit models or check statistical significance.
We can compute the OBD directly from the known population distributions.

\textbf{Two different conclusions.}
\cref{tab:ob_decomp} summarizes the OBD results under each reference group.
For both choices of reference group the explained components are negative, so both choices of reference imply that $H$'s lower average SBP is due in part to its lower average BMI.
However, the sign of the unexplained OBD component flips depending on the choice of the reference group.
When group $H$ is the reference, the negative unexplained component implies that institutional factors further reduce group $H$'s blood pressure relative to group $K$.
However, when group $K$ is the reference, institutional factors appear to actually \emph{increase} group $H$'s average blood pressure relative to group $K$.
Intuitively, because blood pressure in group $K$ is so much more responsive to BMI, group $H$'s lower average BMI appears so good as to ``over-explain'' the difference in SBP; the unexplained component is then forced to be positive to compensate.

In this example, the sign change is in the unexplained component rather than the explained component (cf.\ \cref{sec:real_example}). Since we simulate the data according to group-specific linear models, there is no misspecification. And since we work directly with population quantities, the difference in OBD conclusions is not a product of sampling error.

\begin{table}[H]
\caption{Population Oaxaca–Blinder decomposition of mean SBP differences. All values are in mmHg.}
\label{tab:ob_decomp}
\begin{center}
\small
\begin{tabular}{lrrr}
\toprule
Reference & Explained & Unexplained & Total gap \\
\midrule
$H$ & $-2.0$ & $\textbf{-0.4}$ & $-2.4$ \\
$K$ & $-2.8$ & $\textbf{0.4}$  & $-2.4$ \\
\bottomrule
\end{tabular}
\end{center}
\vskip -0.1in
\end{table}

\section{Theory}
\label{app:theory}

In this section, we present the proof of \cref{thm:unexplained_ordering_sensitivity} and an example showing that the difference between the two decompositions can be made arbitrarily large even when the mean gap is fixed.

\subsection{Proof of \cref{thm:unexplained_ordering_sensitivity}}
\label{sec:sign_flip_proof}
\begin{proof}
Using \cref{def:unexplained_sign_flip}, we want to show that 
\[
\sign\!\left( \mu_K^T \Delta\beta + \Delta\alpha \right)
\neq
\sign\!\left( \mu_H^T \Delta\beta + \Delta\alpha \right)
\iff
\begin{aligned}
&\mu_H^T \Delta\beta \neq \mu_K^T \Delta\beta, \\
&\min\{\mu_H^T \Delta\beta, \mu_K^T \Delta\beta\}
<
-\Delta\alpha
<
\max\{\mu_H^T \Delta\beta, \mu_K^T \Delta\beta\}.
\end{aligned}
\]

$\Rightarrow$: Suppose that $\sign\left( \mu_K^T \Delta\beta + \Delta\alpha \right) \neq \sign\left( \mu_H^T \Delta\beta + \Delta\alpha \right)$. If $\mu_H^T \Delta\beta = \mu_K^T \Delta\beta$, then both $\mu_g^T \Delta\beta + \Delta\alpha$ would have the same sign, which contradicts our assumption of opposite signs. Thus, we must have $\mu_H^T \Delta\beta \neq \mu_K^T \Delta\beta$.

Next, suppose that $\mu_K^T \Delta\beta + \Delta\alpha >0$ and $\mu_H^T \Delta\beta + \Delta\alpha <0$. Then, 
\[
\mu_K^T \Delta\beta > -\Delta\alpha \quad \text{and} \quad \mu_H^T \Delta\beta < -\Delta\alpha.
\]
Combining these two inequalities, we have 
\begin{equation}
  \label{eq:intermediate_ineq1}
  \mu_H^T \Delta\beta < -\Delta\alpha < \mu_K^T \Delta\beta.
\end{equation}
Similarly, if $\mu_K^T \Delta\beta + \Delta\alpha <0$ and $\mu_H^T \Delta\beta + \Delta\alpha >0$, then
\begin{equation}
  \label{eq:intermediate_ineq2}
  \mu_K^T \Delta\beta < -\Delta\alpha < \mu_H^T \Delta\beta.
\end{equation}
Combining \cref{eq:intermediate_ineq1,eq:intermediate_ineq2}, we conclude that
\[
  \min\{\mu_H^T \Delta\beta, \mu_K^T \Delta\beta\} < -\Delta\alpha < \max\{\mu_H^T \Delta\beta, \mu_K^T \Delta\beta\}.
\]

$\Leftarrow$; We now show the converse direction. Suppose that:

$(1)$ $\mu_H^T \Delta\beta \neq \mu_K^T \Delta\beta$ and 

$(2)$ $\min\{\mu_H^T \Delta\beta, \mu_K^T \Delta\beta\} < -\Delta\alpha < \max\{\mu_H^T \Delta\beta, \mu_K^T \Delta\beta\}$. 

We want to show that $\sign\left( \mu_K^T \Delta\beta + \Delta\alpha \right) \neq \sign\left( \mu_H^T \Delta\beta + \Delta\alpha \right)$.
$(2)$ implies that one of the following must be true:
\[
  \mu_H^T \Delta\beta < -\Delta\alpha < \mu_K^T \Delta\beta, \quad \text{or} \quad \mu_K^T \Delta\beta < -\Delta\alpha < \mu_H^T \Delta\beta.
\]
If the first one holds, then $\mu_H^T \Delta\beta + \Delta\alpha <0$ and $\mu_K^T \Delta\beta + \Delta\alpha >0$, that is, the unexplained components have opposite signs. If the second holds, a symmetric argument yields the same conclusion as desired. 


\end{proof}





\subsection{Unbounded sensitivity at fixed mean gap}
 \label{ex:unbounded}
In this section we present a toy example where we can have arbitrarily large flip magnitudes. Consider the following example in dimension \( 1 \), where we assume we have a parameter \( L > 0 \) which controls the magnitude of the covariate shift, and \( \varepsilon > 0 \) which controls the deviation in slope coefficients. We assume the mean vector is given by \( (\mu_K, \mu_H) = (0, L) \), and the slope parameters deviate symmetrically from \( \frac{1}{L} \), so that \( \beta_H = \frac{1}{L} + \frac{\varepsilon}{2} \), \( \beta_K = \frac{1}{L} - \frac{\varepsilon}{2} \). Then, the explained and unexplained components under both reference groups are:
\begin{align*}
    E^{(H)} &= \Delta\mu \cdot \beta_H = L\left(\frac{1}{L} + \frac{\varepsilon}{2}\right) = 1 + \frac{L\varepsilon}{2}, \quad U^{(H)} = \mu_K \cdot \Delta\beta = -\frac{L\varepsilon}{2}, \\
    E^{(K)} &= \Delta\mu \cdot \beta_K = L\left(\frac{1}{L} - \frac{\varepsilon}{2}\right) = 1 - \frac{L\varepsilon}{2}, \quad U^{(K)} = \mu_H \cdot \Delta\beta = +\frac{L\varepsilon}{2}.
\end{align*}

Hence, the difference between the two decompositions is \( E^{(H)} - E^{(K)} = L\varepsilon \), which can be made arbitrarily large as \( L \to \infty \), even though the total mean gap is always fixed at \( \Delta_Y = 1 \). Similarly, the unexplained components differ by \( U^{(K)} - U^{(H)} = L\varepsilon \).

\begin{center}
\begin{tabular}{ccccccc}
\toprule
\(L\) & \(\varepsilon\) 
& \(E^{(H)}\) & \(E^{(K)}\) 
& \(U^{(H)}\) & \(U^{(K)}\) 
& Flip? \\
\midrule
20    & 0.1  & \(2\)     & \(0\)      & \(-1\)    & \(1\)     & \checkmark \\
200   & 0.1  & \(11\)    & \(-9\)     & \(-10\)   & \(10\)    & \checkmark \\
1000  & 0.1  & \(51\)    & \(-49\)    & \(-50\)   & \(50\)    & \checkmark \\
\bottomrule
\end{tabular}
\end{center}

Note that the sign flip is symmetric by construction, but the magnitude of the discrepancy grows linearly in \( L \), highlighting that the attribution gap between reference groups can be stretched arbitrarily far.

\subsection{Decision Tree Representation of \cref{eq:sign_flip_condition}}

To illustrate \cref{thm:unexplained_ordering_sensitivity}, we represent the condition in \cref{eq:sign_flip_condition} as a decision tree in \cref{alg:decision_tree} and use it to scrutinize the healthcare example in \cref{sec:simulated_example}.
In this example, $\mu_H$ and $\mu_K$ denote mean BMIs, both of which are positive by definition and thus $\sign{(\mu_H \Delta\beta)} = \sign{(\mu_K \Delta\beta)} = \sign{(\Delta\alpha)}$. 
Therefore, following \Cref{alg:decision_tree}, a sign flip will occur only if $\sign{(\Delta\alpha)} \neq \sign{(\mu_H \Delta\beta)} = \sign{(\Delta\beta)}$; equivalently, it must hold that one group has a larger $\beta_g$ (more sensitive to higher BMI) but lower $\alpha_g$ (lower SBP on average).
This might occur if one group has better access to treatment but worse genetic factors, as in \cref{sec:simulated_example}.
Furthermore, a sign flip will occur if either $\abs{\mu_H \Delta\beta} < \abs{\Delta\alpha} < \abs{\mu_K \Delta\beta}$ or $\abs{\mu_K \Delta\beta} < \abs{\Delta\alpha} < \abs{\mu_H \Delta\beta}$. This condition implies a range of values over which sign flips will occur; that is, there will be at least one choice of reference group under which improved healthcare appears to help in lowering mean SBP and another ordering under which improved healthcare appears to \emph{hurt} mean SBP.

Even though \cref{thm:unexplained_ordering_sensitivity} does not restrict the magnitude of the parts of the unexplained components, in \cref{ex:unbounded} we provide an example to illustrate that these sign flips are not relegated to settings in which $\mu_g^T\Delta\beta$ are small in magnitude; in fact, we show that they can be arbitrarily large in magnitude.

  \begin{algorithm}
    \caption{Unexplained component sign flip}
    \label{alg:decision_tree}
    \begin{algorithmic}  
        \If{$\mu_H^T \Delta\beta \neq \mu_K^T \Delta\beta$}
            \If{$\sign(\mu_H^T \Delta\beta) = \sign(\mu_K^T \Delta\beta)$}
                \If{$\sign(\Delta\alpha) \neq \sign(\mu_H^T \Delta\beta)$}
                    \If{$|\mu_H^T \Delta\beta| < |\Delta\alpha| < |\mu_K^T \Delta\beta|$ \textbf{or} $|\mu_K^T \Delta\beta| < |\Delta\alpha| < |\mu_H^T \Delta\beta|$}
                        \State A sign flip occurs.
                    \Else
                    	\State There is no sign flip.
                    \EndIf
                \Else
                	\State There is no sign flip.
                \EndIf
            \ElsIf{$\sign(\mu_H^T \Delta\beta) \neq \sign(\mu_K^T \Delta\beta)$}
                \If{$\sign(\Delta\alpha) = \sign(\mu_H^T \Delta\beta)$ \textbf{and} $|\mu_K^T \Delta\beta| > |\Delta\alpha|$}
                    \State A sign flip occurs.
                \ElsIf{$\sign(\Delta\alpha) \neq \sign(\mu_H^T \Delta\beta)$ \textbf{and} $|\mu_H^T \Delta\beta| > |\Delta\alpha|$}
                    \State A sign flip occurs.
                \Else
                    \State There is no sign flip.
                \EndIf
            \EndIf
        \Else
            \State There is no sign flip.
        \EndIf
    \end{algorithmic}
\end{algorithm}

\subsection{Relationship Between Unexplained Component Sign Flips and Omitted Variables}
\label{appendix:ovb}
Here, we assume the same standardization of covariates as in \cref{sec:sign_flip_fraction}; that is, we assume $\mu_K = \mathbf{0}$ and $\mu_H = \mathbf{1}$.
Under this normalization, unexplained component sign flips require \( | \mathbf{1}^T \Delta \beta | > | \Delta \alpha | \), so any data-generating processes that jointly produces large \( | \mathbf{1}^T \Delta \beta | \) and \( | \Delta \alpha | \) would make sign flips less likely than \cref{prop:unexplained_fraction} suggests.
In this section, we present a mechanism that can jointly produce large \( | \mathbf{1}^T \Delta \beta | \) and \( | \Delta \alpha | \) without invoking \cref{ass:no_sign_flip}.

The mechanism involves the researcher's dataset omitting covariates that are related to both \( X \) and \( Y \).
For example, hospital records only contain measurements that doctors have ordered, not the results of important tests a doctor may have skipped.
Similarly, census data used for studying wage gaps typically lack detailed information about skills and experience from a worker's resume and work history.
Finally, in practice, nonlinear transformations of the covariates are also effectively omitted from analysis, like when a researcher does not include appropriate polynomial terms in a regression specification.
Let \( Z \in \mathbb{R}^{d'} \) be a vector of such \textit{omitted covariates}.
We will establish that both \( \Delta \beta \) and \( \Delta \alpha \) can depend on the relationship between \( Z, X, \) and \( Y \).

Our argument requires us to specify the conditional expectations \( \E_g[Y \mid X, Z] \) and \( \E_g[Z \mid X] \).
To establish a simple example, we will assume that both are linear and satisfy
\begin{align}
  \label{eq:y_on_x_z}
  \E_g[Y \mid X, Z] &= \omega_g + X^\top \theta_g + Z^\top \gamma_g, \\
  \label{eq:z_on_x}
  \E_g[Z \mid X] &= \zeta_g + \Psi_g X,
\end{align}
where $\omega_g \in \R$, $\theta_g \in \R^d$, $\gamma_g \in \R^{d'}$, $\zeta_g \in \R^{d'}$, and $\Psi_g \in \R^{d' \times d}$.
Relaxing these assumptions may lead to other non-uniform relationships between \( \Delta \beta \) and \( \Delta \alpha \).

Recall that the OBD is based on regressing the outcome against the observed covariates for each group \(g \in \{H, K\}\) via
\begin{align}
    \label{eq:y_on_x}    
    \E_g[Y \mid X] = \alpha_g + X^\top \beta_g.
\end{align}
We will start by showing that \( \alpha_g \) and \( \beta_g \) depend on \( \gamma_g, \zeta_g, \), and \( \Psi_g \). 
To do so, apply the law of iterated expectations to \cref{eq:y_on_x}, plug in \cref{eq:y_on_x_z,eq:z_on_x}, and then rearrange terms:
\begin{align}    
    \E_g[Y \mid X] 
    &= \E_g[ \E_g[Y \mid X, Z] \mid X ] \nonumber \\
    &= \E_g[ \omega_g + X^\top \theta_g + Z^\top \gamma_g \mid X ] \nonumber  \\
    &= \omega_g + X^\top \theta_g + \E_g[Z \mid X]^\top \gamma_g \nonumber \\
    &= \omega_g + X^\top \theta_g + (\zeta_g + \Psi_g X)^\top \gamma_g \nonumber \\ 
    \label{eq:ovb}
    &= \underbrace{\omega_g + \zeta_g^\top \gamma_g}_{=\alpha_g} + \E_g[X]^\top \underbrace{(\theta_g + \Psi_g^\top \gamma_g)}_{=\beta_g}.
\end{align}
\cref{eq:ovb} is the well-known \textit{omitted variables bias} formula.
It describes how the slope and intercept of a regression equation depend on variables that are not included in the equation.
We can use this formula to express the difference in slopes and intercepts across the two populations as 
\begin{align}  
  \label{eq:delta_alpha_gamma}
  \Delta \alpha &= \Delta \omega + \Delta \zeta^\top \Delta \gamma,  \\
  \label{eq:delta_beta_gamma}
  \Delta \beta  &= \Delta \theta + \Delta \Psi^\top \Delta \gamma,
\end{align}
where \( \Delta \) takes the $H-K$ difference in the quantity that follows; for example, \( \Delta \omega := \omega_H - \omega_K \).
These equations show that \( \Delta \alpha \) and \( \Delta \beta \) depend on the role of omitted covariates in the data-generating process.
Crucially, they provide a mechanism to produce \( | \Delta \alpha | \) and \( | \mathbf{1}^\top \Delta \beta | \) that are jointly large:
if both the Z-X and Z-Y relationships change in complementary ways, then both \( \Delta \zeta^\top \Delta \gamma \) and \( \mathbf{1}^\top \Delta \Psi^\top \Delta \gamma \) will be large. 
If these terms are not counterbalanced by \( \Delta \omega \) and \( \Delta \theta \), the result will be large \( | \Delta \alpha | \) and \( | \mathbf{1}^\top \Delta \beta | \).

\subsection{Volume of the Sign-Flip Parameter Region}
We first show how to choose units for the covariates conveniently, so that \( \Delta \mu = \mathbf{1} \).
This will simplify volume calculations for the regions of parameter space that flip the explained and unexplained components of the OBD.
\begin{lemma}
  \label{lemma:units}
  If \( \mu_H \neq \mu_K \) elementwise, we can choose units for \( X \) such that \( \mu_K := \E_K[X] = \mathbf{0} \) and \( \mu_H := \E_H[X] = \mathbf{1} \).
  \begin{proof}  
  Consider the transformation \( \tilde{X} = (X - \mu_K) / (\mu_H - \mu_K) \), where the division is elementwise.
  This is well-defined as long as \( \mu_H \neq \mu_K \) elementwise.
  Since \( \tilde{X} \) is a linear transformation of \( X \), it corresponds to a change in units.
  Under these new units, \( \E_K[\tilde{X}] = \mathbf{0} \) and \( \E_H[\tilde{X}] = \mathbf{1} \).  
  \end{proof}
\end{lemma}

\subsubsection{Proof of \cref{prop:explained_fraction}}
\label{proof:explained_fraction}
Apply \cref{lemma:units} so that \( \E_K[X] = \mathbf{0} \), \( \E_H[X] = \mathbf{1} \), and therefore \( \Delta \mu = \mathbf{1} \).
\cref{eq:sign_flip_condition_explained_neq} now simplifies to 
\[ \sign\!\left(\mathbf{1}^\top \beta_{H}\right) \neq \sign\!\left(\mathbf{1}^\top \beta_{K}\right). \]
A Uniform$(-M, M)$ distribution on \( \beta_H \) and \( \beta_K \) corresponds to distributions on \( \mathbf{1}^\top \beta_{H} \) and \( \mathbf{1}^\top \beta_{K} \) that are identical and symmetric around zero.
As such, the signs of \( \mathbf{1}^\top \beta_{H} \) and \( \mathbf{1}^\top \beta_{K} \) will disagree in 50\% of the parameter space.

\subsubsection{Proof of \cref{prop:unexplained_fraction}}
\label{proof:unexplained_fraction}
Apply \cref{lemma:units} so that \( \E_K[X] = \mathbf{0} \), \( \E_H[X] = \mathbf{1} \), and therefore \( \Delta \mu = \mathbf{1} \).
For the rescaled covariates, the conditions in \cref{thm:unexplained_ordering_sensitivity} simplify.
In particular, the unexplained component flips sign if and only if \( \mathbf{1}^\top \Delta \beta \neq 0 \), \( \sign (\Delta \alpha) \neq \sign (\mathbf{1}^\top \Delta \beta) \), and \( | \mathbf{1}^\top \Delta \beta | > |\Delta \alpha| \).

To compute the fraction of $C_M$ for which there are sign flips, we equivalently compute the probability of a sign flip under an independent $\mathrm{Uniform}(-M,M)$ measure on each element of the vector $(\beta_H, \beta_K, \alpha_H, \alpha_K) \in \R^{2d+2}$.
The condition \( \mathbf{1}^\top \Delta \beta \neq 0 \) holds almost surely under the assumed Uniform$(-M,M)$ measure, since \( \mathbf{1}^\top \Delta \beta \) is a non-degenerate linear combination of continuously distributed random variables, so we can focus on the latter two conditions.
To do so, we will determine the distributions of  \( \Delta \alpha \) and \( \mathbf{1}^\top \Delta \beta \).

The sum of \( n \) Uniform$(0,1)$ random variables is sometimes called the Irwin-Hall(\(n\)) distribution.%
\footnote{For \(n = 2\) this is a familiar triangular distribution.}
When we place a Uniform$(-M, M)$ measure on the \( \alpha \)'s and \( \beta \)'s, there is a straightforward relationship between \( \Delta \alpha \), \( \mathbf{1}^\top \Delta \beta \), and appropriately defined Irwin-Hall random variables.
The difference \( \Delta \alpha \) is the sum of two Uniform$(-M, M)$ variables, which is equivalent to an Irwin-Hall($2$) distribution, rescaled from $[0,2]$ to $[-2M, 2M]$.
Similarly, \( \mathbf{1}^\top \Delta \beta \) is the sum of \( 2 d \) Uniform$(-M, M)$ variables, and its distribution is an Irwin-Hall($2 d$) distribution, rescaled from $[0,2d]$ to \( [-2dM, 2dM] \).

That is, let \( I_2 \sim \text{Irwin-Hall}(2) \) and \( J_{2d} \sim \text{Irwin-Hall}(2d) \) be independent.
We can represent
\begin{align*}
    &\mathbf{1}^\top\Delta \alpha = 2M (I_2 - 1) 
    &\mathbf{1}^\top\Delta \beta  = 2M (J_{2d} - d)
\end{align*}

We can use this correspondence to translate the event 
\[
\{ |\mathbf{1}^\top \Delta \beta| > |\Delta \alpha| \}
\cap
\{ \sign(\Delta \alpha) \neq \sign(\mathbf{1}^\top \Delta \beta) \}.
\]
into an equivalent event in terms independent Irwin-Hall random variables:
\begin{align*}
    Pr& \left( \{ |\mathbf{1}^\top \Delta \beta| > |\Delta \alpha| \}
    \cap
    \{ \sign(\Delta \alpha) \neq \sign(\mathbf{1}^\top \Delta \beta) \} \right) \\
    =& 
    Pr \left( \{ \Delta \alpha > 0 \} \cap \{ \mathbf{1}^\top \Delta \beta < - \Delta \alpha \} \right)
    +
    Pr \left(
    \{ \Delta \alpha < 0 \} \cap \{ \mathbf{1}^\top\Delta \beta > - \Delta \alpha \} \right) \\
    =& 
    Pr \left( \{ 2M (I_2 - 1)  > 0 \} \cap \{ 2M (J_{2d} - d) < - 2M (I_2 - 1)  \} \right)  \\
    &+
    Pr \left( \{ 2M (I_2 - 1)  < 0 \} \cap \{ 2M (J_{2d} - d) > - 2M (I_2 - 1)  \} \right) \\
    =& 
    Pr \left( \{ I_2 > 1 \} \cap \{ J_{2d} < d + 1 - I_2 \}  \right) +
    Pr \left( \{ I_2 < 1 \} \cap \{ J_{2d} > d + 1 - I_2 \} \right).
\end{align*}

That is, unexplained component sign flips have the same volume as the region where \( I_2 \) falls above its median and \( J_{2d} \) falls sufficiently far in its left tail, or vice versa.

Finally, as \( d \to \infty \), sign flips have a volume equal to
\begin{equation}
\label{eq:irwin_hall_pr}
 Pr \left( \{ I_2 > 1 \} \cap \{ J_{2d} < d \}  \right) +
    Pr \left( \{ I_2 < 1 \} \cap \{ J_{2d} > d \} \right).
\end{equation}
Since the median of \( I_2 \) is $1$ and the median of \( J_{2d} \) is $d$, and the two random variables are independent, each probability equals \( 0.25 \).
The corresponding fraction of $C_M$ leading to sign flips is therefore \( 0.5 \) as claimed.

\begin{remark} \label{remark:computing_irwin_hall}
We can calculate the probabilities in \cref{eq:irwin_hall_pr} exactly for small \( d \) by using the Irwin-Hall$(n)$ cumulative distribution function (CDF),
\begin{equation}
F(x) = \frac{1}{n!} \sum_{k=0}^{\lfloor x \rfloor} (-1)^k \binom{n}{k} (x - k)^n. \label{eq:exact_irwin_hall}
\end{equation}
However, for larger \( d \), the combinatorial terms in the CDF become difficult to evaluate numerically.
In those cases the central limit theorem allows us to approximate an Irwin-Hall$(2d)$ distribution with a \( N(d, \frac{d}{6})\) distribution.
We compute the probabilities in \cref{sec:sign_flip_fraction} for $d \leq 40$ using the exact CDF in \cref{eq:exact_irwin_hall}, and we use the normal approximation for $d > 40$.
\end{remark}

\subsubsection{Volume of Sign Flip Region Without Covariate Standardization}
\label{appendix:sign_flip_simulation}
The preceding analysis calculates the volume of the sign-flip region when covariates are standardized via \cref{lemma:units}.
Here, we repeat the exercise without standardizing any of the covariates.
Let each of \( \alpha_H, \alpha_K, \beta_H, \beta_K, \mu_H \) and \( \mu_K \) lie inside a cube of appropriate dimension, with side length \( 2M > 0 \) and center at the origin.
Deriving a closed form expression for the volume of the sign flip region is challenging, so we estimate this volume by simulation.
For each \( d \in \{1, \dots, 100 \} \) and \( 2M \in \{2, 20, 200\} \), we draw each regression parameter independently from an appropriate Uniform distribution, and then determine whether the unexplained component flips signs.
We simulate 1,000 draws for each combination of \( d \) and \( M \), and plot the proportion of unexplained component sign flips for each combination in \cref{fig:prob_unexplained_simulated}.

For the explained component, the percentage of the parameter space with sign flips appears constant at 50 \%, regardless of \( d \) or \( M \). 
This is consistent with the result for covariates standardized via \cref{lemma:units}.
For the unexplained component, the percentage of the parameter space with sign flips is nonzero for even small \( d \) and \( M \), and appears to approach \( 50\% \) as each of \( d \) and \( M \) grow.
This behavior is similar to the result for standardized covariates.
Together, these findings suggest that the conclusions from \cref{proof:explained_fraction,prop:unexplained_fraction} are not simply a result of our choice of units.

\begin{figure}[H]
    \centerline{\includegraphics[width=0.5\linewidth]{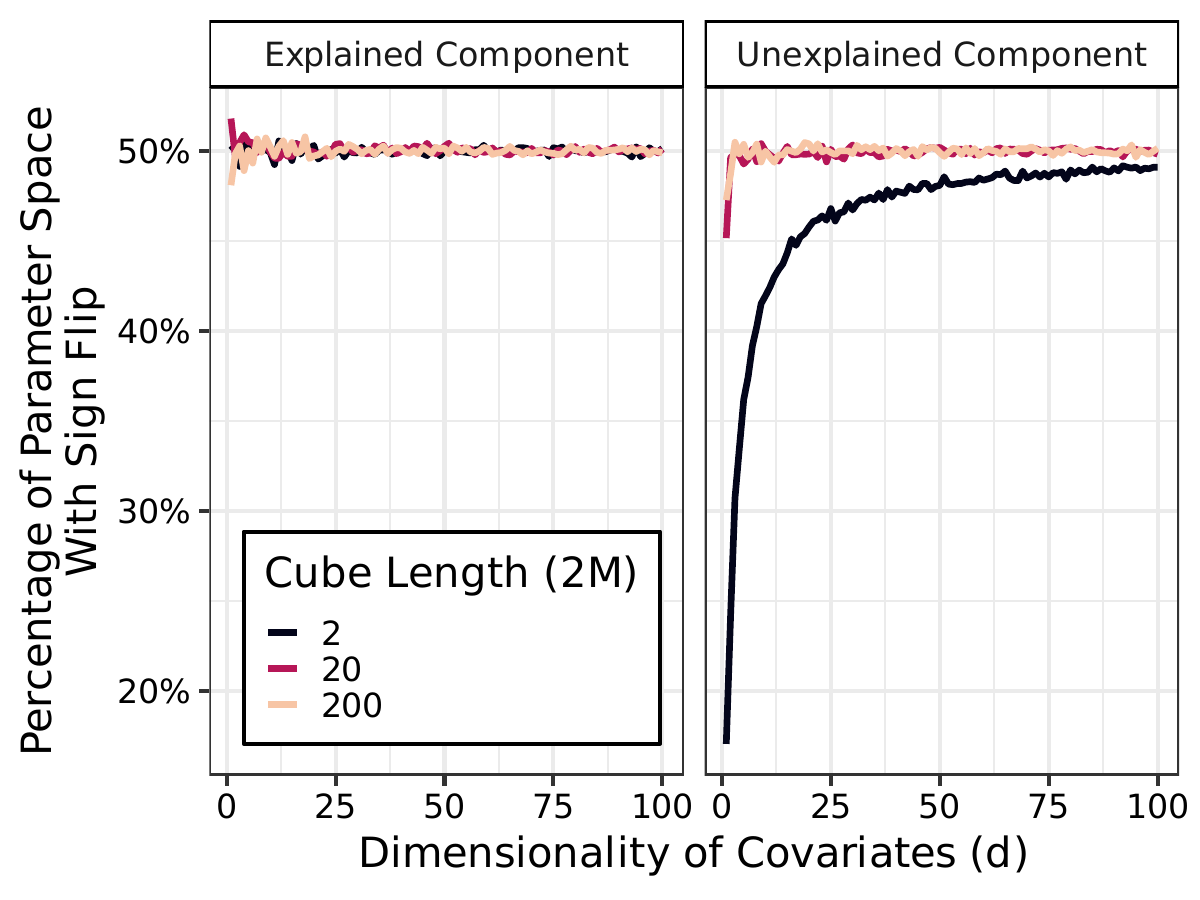}}
    \caption{Percentage of parameter space leading to unexplained component sign flips, without covariate standardization.}
    \label{fig:prob_unexplained_simulated}
\end{figure}

\subsection{Proofs from \cref{sec:reconciling_theory_empirics}}

\subsubsection{Proof of \cref{prop:no_sign_flip}}
\label{proof:no_sign_flip}
Recall that by definition
$$
	U^{(K)} := \mu_H^T \Delta\beta + \Delta\alpha, \quad \quad U^{(H)} := \mu_K^T \Delta\beta + \Delta\alpha.
$$
By assumption, $\sign(\Delta\alpha) = \sign(\mu_H^T\Delta\beta) = \sign(\mu_K^T\Delta\beta)$, and thus $\sign(U^{(H)}) = \sign(U^{(K)})$.


\section{Licenses of software assets}
\label{appendix:licenses}

Below we list all libraries and packages used in our code, along with the version pinned during our experiments and the corresponding license.

\begin{itemize}
    \item \texttt{numpy==1.26.*} --- BSD 3-Clause License
    \item \texttt{pandas==2.2.*} --- BSD 3-Clause License
    \item \texttt{scipy==1.13.*} --- BSD 3-Clause License
    \item \texttt{scikit-learn==1.5.*} --- BSD 3-Clause License
    \item \texttt{statsmodels==0.14.*} --- BSD 3-Clause License
    \item \texttt{matplotlib==3.8.*} --- BSD 3-Clause License
    \item \texttt{xgboost==2.1.*} --- Apache License 2.0
    \item \texttt{joblib==1.4.*} --- BSD 3-Clause License
    \item \texttt{tabpfn==2.0.*} --- Custom license: \url{https://huggingface.co/Prior-Labs/tabpfn_2_6/blob/main/LICENSE}
\end{itemize}

Our code also uses R version 4.5.2 (License: GNU GPL version 2); the replication code lists all used R software packages, versions, and licenses in the `renv.lock' file.


\newpage

\end{document}